\documentclass[3p,preprint]{elsarticle}

\journal{Theoretical Computer Science}

\usepackage{amssymb}
\usepackage[figuresright]{rotating}

\newtheorem{theorem}{Theorem}
\newtheorem{lemma}[theorem]{Lemma}
\newtheorem{proposition}[theorem]{Proposition}
\newtheorem{corollary}[theorem]{Corollary}
\newtheorem{definition}[theorem]{Definition}

\newtheorem{example}{Example}

\newproof{proof}{Proof}
\usepackage[utf8]{inputenc}
\usepackage[algosection]{algorithm2e}
\usepackage{endnotes}
\usepackage{listings}
\usepackage{mathtools}
\usepackage{verbatim}
\usepackage{mathrsfs}
\usepackage{shuffle}
\usepackage{longtable}
\usepackage{enumitem}
\usepackage{etoolbox}
\usepackage{float}
\usepackage{color}
\usepackage{hyperref}

\usepackage{complexity}

\newclass{\NQA}{NQA}
\newclass{\NSA}{NSA}
\newclass{\NFPA}{NFPA}

\newclass{\DCM}{DCM}
\newclass{\eDCM}{eDCM}
\newclass{\eNPDA}{eNPDA}
\newclass{\DPDA}{DPDA}
\newclass{\PDA}{PDA}
\newclass{\DCMNE}{DCM_{NE}}
\newclass{\TwoDCM}{2DCM}
\newclass{\NCM}{NCM}
\newclass{\eNCM}{eNCM}
\newclass{\eNQA}{eNQA}
\newclass{\eNSA}{eNSA}
\newclass{\eNPCM}{eNPCM}
\newclass{\eNQCM}{eNQCM}
\newclass{\eNSCM}{eNSCM}
\newclass{\DPCM}{DPCM}
\newclass{\NPCM}{NPCM}
\newclass{\NQCM}{NQCM}
\newclass{\DQCM}{DQCM}
\newclass{\NSCM}{NSCM}
\newclass{\NPDA}{NPDA}
\newclass{\TRE}{TRE}
\newclass{\NFA}{NFA}
\newclass{\DFA}{DFA}
\newclass{\NCA}{NCA}
\newclass{\DCA}{DCA}
\newclass{\DTM}{DTM}
\newclass{\NTM}{NTM}
\newclass{\DLOG}{DLOG}

\newcommand{\LFam}{{\cal L}}
\newcommand{\MFam}{{\cal M}}
\newcommand{\SFam}{{\cal S}}

\newsavebox{\spacebox}
\begin{lrbox}{\spacebox}
\verb*! !
\end{lrbox}
\newcommand{\blank}{\usebox{\spacebox}}%
\newcommand\rw{\downarrow}

\newcommand\acc{Acc}
\newcommand\coacc{co\mbox{-}Acc}

\begin{document}

\begin{frontmatter}

\title{On Store Languages of Language Acceptors\tnoteref{t1}}

\tnotetext[t1]{\textcopyright 2018. This manuscript version is made available under the CC-BY-NC-ND 4.0 license \url{http://creativecommons.org/licenses/by-nc-nd/4.0/}}

\author[label1]{Oscar H. Ibarra\fnref{fn1}}
\address[label1]{Department of Computer Science\\ University of California, Santa Barbara, CA 93106, USA}
\ead[label1]{ibarra@cs.ucsb.edu}
\fntext[fn1]{Supported, in part, by
NSF Grant CCF-1117708 (Oscar H. Ibarra).}

\author[label2]{Ian McQuillan\fnref{fn2}}
\address[label2]{Department of Computer Science, University of Saskatchewan\\
Saskatoon, SK S7N 5A9, Canada}
\ead[label2]{mcquillan@cs.usask.ca}
%\cortext[corr]{Corresponding author}
\fntext[fn2]{Supported, in part, by Natural Sciences and Engineering Research Council of Canada Grant 2016-06172 (Ian McQuillan).}

\begin{abstract}
It is well known that the ``store language'' of every pushdown 
automaton --- the set of store
configurations (state and stack contents) that can appear as an
intermediate step in accepting
computations --- is a regular language.  Here
many models of language acceptors with various
store structures are examined, along with a study of their store
languages. 
%For example, in a pushdown automaton with multiple counters, the
%storage configuration will now consist of the state, stack contents, and
%the values (in unary notation) of the counters.  
For each model, an attempt is made 
to find the simplest model
%describe the complexity of the store languages.
that accepts their store languages.   
Some connections between store languages of one-way and two-way machines are demonstrated, as with connections between nondeterministic and deterministic machines.
A
nice application of these store language results is also presented, showing a general technique for proving families accepted by many deterministic models are closed under right quotient with regular languages, resolving some open questions (and significantly simplifying proofs for others that are known) in the literature.
Lower bounds on the space complexity of Turing machines for having non-regular store 
languages are obtained.
 \end{abstract}

\begin{keyword}
Store Languages \sep Turing Machines \sep Storage Structures \sep Right Quotient \sep Automata
\end{keyword}

\end{frontmatter}

\section{Introduction}

A store configuration of a one-way or two-way language acceptor consists
of the state followed by the contents of its memory (store) structure.
It does not include the input and the position of the input
head.  For example, for a nondeterministic pushdown automaton
($\NPDA$), a store configuration is represented by a string $qx$,
where $q$ is a state and $x$ is the contents of the pushdown
stack.  For multi-tape acceptors, such as for an $\NPDA$ augmented
with $k$ reversal-bounded counters
($\NPCM$) \cite{Ibarra1978}, the store configuration is represented  
by the string $qx c_1^{j_1} \cdots c_k^{j_k}$,
where $j_i$ represents the value of counter $i$ in unary notation,
and the $c_i$ symbols and the symbols of $x$ are disjoint.
For a machine $M$, let $S(M)$ be the set of store
configurations that can appear as an intermediate step in accepting 
computations of $M$.

It is well-known that $S(M)$ is a regular language for any $\NPDA$
$M$ \cite{CFHandbook,GreibachCFStore}. Greibach used this result to provide an alternative
proof \cite{GreibachCFStore} that regular canonical systems produce regular languages \cite{Buchi1990}. Also, it was a key component to showing that it is decidable 
whether the set of all infixes (subwords) of the language accepted by a reversal-bounded\footnote{Reversal-bounded means that there is a bound on the number of switches between non-decreasing and non-increasing the size of the pushdown.} $\NPDA$ is equal to $\Sigma^*$ (i.e., is dense) \cite{DLTJournalIJFCS}.  
Connections between store languages and the area of verification and model checking
have also been recently explored \cite{StoreApplications}.

Due to the usefulness of the store language concept, the
store languages of several models of
language acceptors are studied in this paper.  
For machine models with  an undecidable emptiness problem, membership in the store
language is undecidable. Hence,
the investigation of store languages is particularly focused on
machine models with a decidable emptiness problem.
Results are given here that generalize (in often non-obvious ways)
the aforementioned result concerning $\NPDA$s to many other machine models,
such as the following:
\begin{enumerate}
\item The following nondeterministic machine models with one-way read-only input have regular store languages: $k$-flip pushdown automata \cite{flipPushdown} (which
are like pushdown automata but can flip the pushdown 
store up to $k$ times), reversal-bounded queue automata, nondeterministic Turing machines with a reversal-bounded worktape, and
stack automata \cite{Stack2,StackAutomata}. The result for stack automata was shown recently \cite{KutribCIAA2016} and so our result 
becomes an alternate proof that follows from
existing results in the literature.  Also, a new simple but general
method is presented for translating results between two-way machines
and one-way machines.
%\item
%The store language of a nondeterministic pushdown automaton
%with reversal-bounded counters ($\NPCM$) can be accepted by a one-way
%nondeterministic finite automaton with reversal-bounded counters ($\NCM$).
\item
The store languages of finite-crossing\footnote{Finite-crossing means that the input head crosses the boundary of
any two adjacent input symbols at most a fixed number of times.} two-way nondeterministic machines
with reversal-bounded counters can be accepted by one-way deterministic machines with
reversal-bounded counters ($\DCM$).
\item
There is a non-finite-crossing two-way deterministic
machine with one reversal-bounded counter whose store
language cannot be accepted by any $\NPCM$.
\item Some machine models (e.g., deterministic pushdown automata with reversal-bounded counters, $\DPCM$) cannot accept their own store languages.
\end{enumerate}
$\NPCM$s and $\NCM$s have been extensively studied
since their introductions in \cite{Ibarra1978, Baker1974}.
They have found applications in areas such as timed automata
\cite{Dang2001}, model-checking and verification \cite{HagueLin2011,IBARRA2002165},
membrane computing \cite{PaunBook}, and Diophantine equations \cite{IbarraDang}.

Another interesting application is presented here
showing the closure of many families of languages accepted by
deterministic machines under right quotient with regular languages.
Some of these resolve open problems in the literature,
and others simplify existing known proofs. These
include deterministic stack automata (known with a lengthy proof in \cite{DetStackQuotient}), deterministic $k$-flip pushdown
automata (stated as an unresolved open problem in \cite{kflipquotient}), certain types of deterministic Turing machines, deterministic checking stack automata, and deterministic reversal-bounded queue automata. An alternate proof
of the result for deterministic pushdown automata that was
shown in \cite{GinsburgDPDAs} is also given.
This general closure is somewhat surprising given the determinism of the machines and the nondeterministic nature of deletion occurring with quotients. 

Finally, lower bounds are obtained on the space complexity of different types of Turing machines in order to have non-regular store languages.

\section{Notation}

An alphabet $\Sigma$ is a set of symbols (usually assumed to be finite
unless stated otherwise). The set of all words over $\Sigma$ is
denoted by $\Sigma^*$, and the set of all non-empty words is denoted by $\Sigma^+$.
A {\em language} $L$ over $\Sigma$ is any subset of $\Sigma^*$.
Given a word $w \in \Sigma^*$, the {\em length} of $w$ is denoted by $|w|$.
Given $a \in \Sigma$, then $|w|_a$ is the number of $a$'s in $w$.
The {\em empty word} is denoted by $\epsilon$. 
The {\em reverse} of a word $w$ is denoted by $w^R$, extended to the
reverse $L^R$ of a language $L$ in the natural way.
Given two languages $L_1,L_2$, the {\em left quotient} of $L_2$ by $L_1$, $L_1^{-1}L_2 = \{ y \mid xy \in L_2, x \in L_1\}$, and the {\em right quotient} of $L_1$ by $L_2$ is $L_1 L_2^{-1} = \{x \mid xy \in L_1, y \in L_2\}$.
A language $L\subseteq \Sigma^*$ is {\em letter-bounded} if there exists
(not necessarily distinct) $a_1, \ldots, a_l \in \Sigma$ such that
$L \subseteq a_1^* \cdots a_l^*$. A language $L$ is {\em bounded} if there
exists $w_1, \ldots, w_l \in \Sigma^*$ such that $L\subseteq w_1^* \cdots w_l^*$.
Given two words $u,v \in \Sigma^*$, $u$ is a {\em prefix} of $v$ if $v = ux$,
for some $x \in \Sigma^*$, $u$ is a {\em suffix} of $v$ if $v=xu$ for some 
$x \in \Sigma^*$, $u$ is an {\em infix} of $v$ if $v= xuy$, for some $x,y \in \Sigma^*$, and $u$ is a {\em subsequence} of $v$ if $v = x_0 u_1 x_1 \dots x_{n-1} u_n x_n, x_0, \ldots, x_n, u_1, \ldots, u_n \in \Sigma^*, u = u_1 \cdots u_n$.

In this paper, introductory knowledge of automata and 
formal language theory is assumed (see \cite{HU} for an introduction), including finite automata ($\NFA$s and $\DFA$s), pushdown
automata ($\NPDA$s), Turing machines ($\NTM$s and $\DTM$s), and
generalized sequential machines (gsms).
Let $\LFam(\REG)$ be the family of languages accepted by $\NFA$s.

\section{Store Languages of One-Way Machines}
\label{sec:oneway}

Many different kinds of machine models are studied in this paper, such as finite automata, pushdown automata \cite{HU}, reversal-bounded multicounter
machines \cite{Ibarra1978}, stack automata (similar to a pushdown automata with the ability to read, but not change on the inside of the pushdown) \cite{Stack2,StackAutomata}, Turing machines \cite{HU}, queue automata \cite{Harju2002278}, flip-pushdown automata (machines with the ability to flip the pushdown at most $k$ times) \cite{flipPushdown}, and 
also combinations of their stores within individual machines. The store language of each depends on 
the precise definition of each type of machine. It is possible to define all such models generally by varying 
the ``store type'' similar to Abstract Families of Automata \cite{G75} or storage types \cite{ENGELFRIET}, and then the store language only 
needs to be defined once for all types of machines. A similar approach is followed here due to the large number of machine models
considered, because it allows to make general connections between types of machines, and because store languages depend considerably on the precise
definition of the machines.

\begin{definition}
A {store type} is a tuple $\Omega = (\Gamma, I,f,g,c_0, L_{I})$, where
\begin{itemize}
\item $\Gamma$ is the set of store symbols (potentially infinite, available to all machines using this store),
\item $I$ is the set of instructions,
\item $f$ is the write function, a partial function from $\Gamma^* \times I$ to $\Gamma^*$,
\item $g$ is the read function, a partial function from $\Gamma^*$ to
$\Gamma$,% such that $g([\{\epsilon\}]^k) = [\{\epsilon\}]^k$, and for all $j$, $1 \leq j \leq k$, $g(\gamma(j)) = \epsilon$ if and only if $\gamma(j) = \epsilon$,
\item $c_0 \in \Gamma^*$ is the initial store configuration,
\item $L_{I} \subseteq I^*$ is the instruction language.
\end{itemize}
\end{definition}
Thus, a store type defines a type of auxiliary store. 
The write function $f$ indicates how each store contents change
in response to each instruction, and the read
function $g$ indicates how machines read from each store contents.
Every machine using this store type starts with $c_0$ on its store.
Lastly, $L_I$ is a type of filter that can restrict the allowable
sequences of instructions. This is useful for several purposes, such
as defining reversal-bounded store types.

%One example of a classical store type is given next, followed by the definition of the set of all machines using store types.

\begin{example}
\label{pushdownstore}
The {\em pushdown} store type is 
$\Omega = (\Gamma, I,f,g,c_0, L_{I})$ where 
$c_0 = Z_0 \in \Gamma$, 
$\Gamma_0 = \Gamma - \{Z_0\}$ ($Z_0$ is the bottom-of-stack marker), 
$I = \Gamma^*$, $L_{I} = I^*$ (i.e.\ there is no restriction as to the possible sequences of instructions),
$g(xa)=a, x \in \Gamma^*, a \in \Gamma$, 
$f(xa,y) = xy$, where $y \in \Gamma^*, xa, xy \in Z_0 \Gamma_0^*$.
\end{example}
Intuitively, a pushdown store can read the rightmost symbol,
and can replace the rightmost symbol with any word; these
words are the instructions (which can be the empty word for popping).

The machines (defined next) using this store type are equivalent to standard pushdown automata \cite{HU}.

%Let
%$\lhd$ be a special symbol that is the right input
%end-marker used by all one-way machine models.
%The end-marker is not needed for any nondeterministic model and we will 
%leave it off in those cases, but is needed for some deterministic models.
\begin{definition}
Given store types $\Omega_1, \ldots, \Omega_k$ with
$\Omega_i = (\Gamma_i,I_i,f_i,g_i,c_{0,i}, L_{I,i}), 1 \leq i \leq k$, where $\Gamma_i$ are pairwise disjoint, for $1\leq i \leq k$, a 
one-way nondeterministic $(\Omega_1, \ldots,\Omega_k)$-machine is a tuple
$M = (Q,\Sigma, \Gamma, \delta, q_0, F)$, where $Q$ is the
finite set of states, $\Sigma$ is the input alphabet,
$q_0 \in Q$ is the initial state, $F\subseteq Q$ is the set of final states, 
$\Gamma$ is a finite subset of $\Gamma_1 \cup \cdots \cup \Gamma_k$,
and
the finite transition relation $\delta$ is from 
$Q \times (\Sigma \cup \{\epsilon\}) \times \Gamma_1 \times \cdots \times \Gamma_k$ to
$Q \times I_1 \times \cdots \times I_k$.

A configuration of $M$ is a tuple $(q,w, \gamma_1, \ldots, \gamma_k)$, where
$q\in Q$ is the current state, $w \in \Sigma^*$ is the remaining input, 
and $\gamma_i \in \Gamma_i^*$ is the contents of the $i$'th store, for $1 \leq i \leq k$. 
The derivation
relation $\vdash_M$ is defined by: 
$(q,aw,\gamma_1, \ldots, \gamma_k) \vdash_M (q', w, \gamma_1', \ldots, \gamma_k'), 
w \in \Sigma^*,
a\in \Sigma \cup \{\epsilon\}, \gamma_i,\gamma_i' \in \Gamma_i^*, 1 \leq i \leq k,
q,q' \in Q$, if there exists 
\begin{equation}(q',\iota_1, \ldots,\iota_k ) \in \delta(q,a,d_1, \ldots, d_k), \label{transition}
\end{equation}
such
that $g_i(\gamma_i) = d_i, f(\gamma_i,\iota_i) = \gamma_i'$, for all $i$, $1 \leq i \leq k$.
This is extended to $\vdash_M^*$, the reflexive and transitive closure of $\vdash_M$.
Sometimes, bijective labels $T$ will be associated with transitions
of $M$, and in such cases, the derivation relation using
transition $t$ is sometimes written as $\vdash_M^t, t \in T$, 
generalized to
words $\vdash_M^x, x \in T^*$. For $1 \leq i \leq k$, define
a homomorphism $\pi_i$ from $T^*$ to $I_i^*$ where
$\pi_i(t) = \iota_i$ for $t$ of the form of (\ref{transition}).
Then $x \in T^*$ is valid if $\pi_i(x) \in L_{I,i},$ for each $i$, $1 \leq i \leq k$.

The {\em language accepted by $M$}, $L(M) = \{w \mid (q_0,w, c_{0,1}, \ldots, c_{0,k}) \vdash^x_M 
(q_f, \epsilon, \gamma_1, \ldots, \gamma_k), q_f \in F, w \in \Sigma^*, \gamma_i \in \Gamma_i, 1 \leq i \leq k, x \in T^* \mbox{~is valid}\}$. 

\begin{comment}
The set $\acc_M(q)$ of state $q$ reachable store configurations of $M$ is the set
$$\{\gamma_1  \cdots  \gamma_k \mid (q_0, w , c_{0,1}, \ldots, c_{0,k}) \vdash_M^* (q, \epsilon, \gamma_1, \ldots, \gamma_k ), w \in \Sigma^*,  \gamma_i \in \Gamma_i, 1 \leq i \leq k \},$$ and the set
$\coacc_M(q)$ of state $q$ co-reachable store configurations is the set
$$\{\gamma_1 \cdots \gamma_k \mid (q, w , \gamma_1, \ldots, \gamma_k) \vdash_M^* (q_f, \epsilon, \gamma_1', \ldots, \gamma_k'), 
q_f \in F, w \in \Sigma^*, \gamma_i, \gamma_i' \in \Gamma_i, 1 \leq i \leq k\}.$$ 
\end{comment}
The {\em store language} of $M$, 
$$S(M) = \{q\gamma_1 \cdots \gamma_k \mid \begin{array}[t]{l}
(q_0, uv , c_{0,1}, \ldots, c_{0,k}) \vdash_M^{x} (q, v, \gamma_1, \ldots, \gamma_k )\vdash_M^{y} (q_f, \epsilon, \gamma_1', \ldots, \gamma_k'), \\
q_f \in F, u,v \in \Sigma^*, \gamma_i, \gamma_i' \in \Gamma_i, 1 \leq i \leq k, x,y \in T^*, xy \mbox{~is valid}\}.\end{array}$$
\end{definition}
Thus, $S(M)$ is the set of store configuration representatives that can appear as an intermediate step of an accepting computation. 
It is also enforced that, if $k>1$, then $\Gamma_1, \ldots, \Gamma_k$ are all disjoint. By a slight abuse 
of notation, machines with several tapes that have the same store type are
assumed to have disjoint tape alphabets.
Thus, since
the letters used in each component are disjoint, when reading
a string of $S(M)$, it is possible to know which of the $k$ stores
is being read.

\begin{definition}
For a given set of machines $\MFam$, let $\LFam(\MFam)$ be the family of languages accepted by machines in $\MFam$, and
$\SFam(\MFam)$ be the family of
store languages of machines in $\MFam$.
\end{definition}

Various types of one-way deterministic automata will also be studied
in this paper. These  machines are defined to scan input $w \lhd$, where $w \in \Sigma^*$, and $\lhd$ is the right end-marker
(this is needed for some types of machines such as $\DCM$ \cite{eDCM}). Then, a machine is deterministic if
$|\delta(q,a,d_1, \ldots, d_k) \cup \delta(q,\epsilon, d_1, \ldots, d_k)| \leq 1$ for all
$q \in Q, a \in \Sigma \cup \{\lhd\} , d_i \in \Gamma_i$, the language accepted by $M$, 
$L(M) = \{w \mid (q_0,w \lhd, c_{0,1}, \ldots, c_{0,k}) \vdash^x_M 
(q_f, \epsilon, \gamma_1, \ldots, \gamma_k), q_f \in F, w \in \Sigma^*, \gamma_i \in \Gamma_i^*, 1 \leq i \leq k, x \mbox{~is valid}\}$, and the store language of $M$,
$S(M) = \{q \gamma_1 \cdots \gamma_k \mid (q_0,w \lhd, c_{0,1}, \ldots, c_{0,k}) \vdash^x_M (q, w', \gamma_1, \ldots, \gamma_k) \vdash_M^y
(q_f, \epsilon, \gamma_1', \ldots, \gamma_k'), q_f \in F, w \in \Sigma^*, w' \in \Sigma^*\lhd \cup \{\epsilon\} \gamma_i, \gamma_i' \in \Gamma_i^*, 1 \leq i \leq k, x,y \in T^*, xy \mbox{~is valid}\}$.

Given store types $\Omega_1, \ldots, \Omega_k$, the set of all one-way
nondeterministic, or deterministic, 
$(\Omega_1, \ldots, \Omega_k)$-machines that can be built using this store type will be examined.

%A set of machines
%that contains every every one-way nondeterministic machine (resp.\ one-way deterministic) 
%that can be built using the store types used is called {\em one-way nondeterministic (resp.\ deterministic) total}.

In this notation, store types that are equivalent to standard automata models from the literature will be described. 
Indeed, pushdown automata are machines that start with $Z_0$ on the
pushdown, can read the top of the pushdown on its transitions, and
replace the topmost letter with a word. These correspond with $\Omega$-machines as built with the pushdown store type of Example \ref{pushdownstore}.
Let $\NPDA$ be the set of all such pushdown automata.

An $l$-reversal-bounded pushdown store is the same except $L_{I}$ is set to the concatenation of $l$ alternating sequences of
$(\{y \mid y \in \Gamma^*, |y| \geq 1\})^*$ and $(\{y \mid y \in \Gamma^*, |y| \leq 1\})^*$ i.e.\ there are at most $l$ alternations between
non-decreasing and non-increasing the size of the stack.

A counter store type restricts a pushdown store to having a single symbol $c \in \Gamma_0$ (plus $Z_0$). 
At each step, essentially based on whether
the counter is empty or non-empty, a machine can change each counter by $+1,0$, or $-1$.
Similarly, $l$-reversal-bounded counters can be defined exactly like $l$-reversal-bounded
pushdowns.

One can consider machines with $k$ $l$-reversal-bounded counters. The set of all machines that have $k$ $l$-reversal-bounded counters,
for some $k,l \geq 1$ is denoted by $\NCM$. 
Note that $\NCM$ is a union of sets of machines that can be built using store types.%; i.e.\ all $k$ $l$-reversal-bounded can be built using
%a $k$-tape store.

%Also, the store type of one unrestricted pushdown plus
%$k$ $l$-reversal-bounded counters. The set of automata which include some number of $k$ $l$-reversal-bounded counters is denoted by
%$\NPCM$. 

\begin{example}
A queue store type is a tuple $\Omega=(\Gamma,I,f,g,c_0, L_{I})$, where
$I = \{{\rm enqueue}(y) \mid y \in \Gamma^*\} \cup \{{\rm dequeue}\}$,
$c_0 = \epsilon$, $L_{I} = I^*$, $g(x)$ is $\epsilon$ if $x = \epsilon$, and the leftmost symbol of $x$ otherwise, $f(x,{\rm dequeue}) = \Gamma^{-1} x$,  and
$f(x,{\rm enqueue}(y)) = xy$.
\end{example}
%Like with $\NPCM$, we can define machines with one queue for
%the first tape, and $k$ additional counters. Let $\NQCM$ be the set
%of machines with one queue, and some number $k$ of counters
%where the counters are reversal-bounded.

\begin{example}
\label{flipdefinition}
The $k$-flip-pushdown store type is a tuple
$\Omega = (\Gamma,I,f,g,c_0, L_{I})$, where $c_0 = Z_0$,
$\Gamma_0 = \Gamma - \{Z_0\}$,
$I = \Gamma^* \cup \{{\rm flip}\}$, 
$g(xa)=a, x \in \Gamma^*,a \in \Gamma$, $f(xa,y) = xy$,
where $y \in \Gamma^*, xa, xy \in Z_0 \Gamma_0^*$, and
$f(x,{\rm flip}) = Z_0(Z_0^{-1}x)^R,  x \in \Gamma^*$ (the pushdown above the end-marker flips), and
$L_{I}$ restricts at most $k$ flip moves to be applied.

For example, consider a machine $M = (Q, \Sigma, \Gamma,\delta,q_0,F),
\Sigma = \{a,b ,\$\}, F = \{q_2\}$, with transitions
\begin{align*}
& \delta(q_0,c,d) = \{(q_0,dc) \}, \forall d \in \{Z_0,a,b \}, c \in \{a,b\}, && \delta(q_0,\$, c) = \{(q_1,{\rm flip}) \}, \forall c \in \{a,b\},\\
& \delta(q_1,c,c) = \{(q_1, \epsilon)\}, \forall c\in \{a,b\}, &&  \delta(q_1, \epsilon, Z_0) = \{(q_2,Z_0)\}.
\end{align*}
In every accepting computation, $M$ must push all contents $w$ onto the stack until $\$$, making a stack of $Z_0w$,
where a flip occurs producing $Z_0 w^R$. Then the contents are popped and matched to the input. Thus,
$L(M) = \{w \$ w \mid w \in \{a,b\}^+\}$.
\end{example}
Machines defined with this type are equivalent to those
in \cite{flipPushdown}, although they are classically defined
where the flips are performed with a separate function.

%\begin{comment}
Next, we define stacks similarly to the definition in \cite{Stack2,G75}.
They are defined like pushdowns, but there are also instructions
to enter the inside of the stack in a two-way read-only fashion.
\begin{example}
The stack store type is a tuple 
$\Omega = (\Gamma, I,f,g,c_0, L_I)$,
where 
$\Gamma$ is an infinite set of store
symbols available to stacks, with special symbols $\downarrow \in \Gamma$
that gives the position of the read/write head in the stack, $Z_b\in \Gamma$ is the bottom-of-stack marker, and $Z_t \in \Gamma$ is the top-of-stack marker, with
$\Gamma_0 = \Gamma - \{\downarrow\}$, $\Gamma_1 = \Gamma_0 - \{Z_b,Z_t\}$,
$I = \Gamma_0^* \cup \{{\rm D}, {\rm S}, {\rm U}\}$
is the set of instructions of the stack, where
the first set are changing the top symbol of the stack, 
and the rest (down, stay, or up) move the read/write head inside the stack,
$L_I = I^*$, $c_0 = Z_b \downarrow Z_t$,
$g( x a \downarrow  x') = a, 
a \in \Gamma_0, x,x'\in \Gamma_0^*$ with 
$xax' \in Z_b \Gamma_1^* Z_t$,
and $f$ is defined as:
\begin{itemize}
\item $f(x a\downarrow Z_t, y) = x y \downarrow Z_t$ for $x,y \in \Gamma_0^*,a \in \Gamma_0,xa,xy \in Z_b \Gamma_1^*$,
\item $f(Z_b x a \downarrow x'   , {\rm D}) = Z_b x \downarrow a x'$,
for $x, x' \in \Gamma_0^*, a \in \Gamma_1 \cup \{Z_t\}$, with $xax' \in \Gamma_1^* Z_t$,
\item $f(Z_b x \downarrow x', {\rm S}) = Z_b x\downarrow x'$, for
$x,x' \in \Gamma_0^*, x x' \in \Gamma_1^* Z_t$,
\item $f(Z_b x \downarrow a x', {\rm U}) = Z_b x a \downarrow x'$,
for $x,x' \in \Gamma_0^*, a \in \Gamma_1 \cup \{Z_t\}, xax' \in \Gamma_1^* Z_t$.
\end{itemize}
Also, the checking stack store type is a restriction of stack store type above where $L_I$ is restricted to be in 
 $y \mid y \in \Gamma_0^+\}^* \{{\rm D}, {\rm S}, {\rm U}\}^*$. 
 That is, a checking stack has two phases, a ``writing phase'',
 where it can push or stay (no pop), and then a ``reading phase'', where it
 enters the stack in read-only mode. But once it starts reading, it cannot change the stack again. 
\end{example}
%\end{comment}
Stacks require that the 
read/write head be included in the store language in a similar fashion to Turing tapes,
as defined next:
%Like with $\NPCM$, we can also define machines with one stack for
%the first tape, and $k$ of additional counters. Let $\NSCM$ be the set
%of machines with one stack, and some number $k$ of counters
%where the counters are reversal-bounded.

\begin{example}
A Turing store is a tuple $\Omega = (\Gamma,I,f,g,c_0, L_{I})$
where $c_0 = \blank\rw$ ($\rw$ is the read/write head that reads the symbol before it, and $\blank$
is the blank symbol, so each tape initially only has a blank followed by the read/write head, $\Gamma_0 = \Gamma - \{\downarrow\},
\Gamma_1 = \Gamma_0 - \{\blank\}$), 
$I = \{a^{\leftarrow}, a^{\rightarrow}, a \mid a \in \Gamma_0\}$ (this groups together both the new symbol written in the current tape cell, and the direction for the head to move), $L_{I} = I^*$,
$g(x b \rw x') = b, b \in \Gamma_0, x,x' \in \Gamma_0^*$, $x$ does not start with $\blank$, and
$x'$ does not end with $\blank$. 
At each step, a machine with this store can read the symbol under
the read/write head, and execute an instruction which 
corresponds to a standard Turing machine instruction, writing
an $a$ in the current cell and moving left, right, or staying. 
The write function is defined by, for all $x \in \Gamma_1 \Gamma_0^* \cup \{\epsilon\}, x' \in \Gamma_0^* \Gamma_1 \cup \{\epsilon\}, a,b \in \Gamma_0$:
\begin{itemize}
\item $f(x b\rw x', a) = x a\rw  x'$,
\item $f(x b \rw x', a^{\leftarrow}) = 
\begin{cases}
\blank \rw a x' & \mbox{~if~} a x' \notin \blank^* \mbox{~and~} x=\epsilon,\\ 
x \rw a x' & \mbox{~if~} a x' \notin \blank^* \mbox{~and~} x\neq \epsilon,\\ 
\blank \rw & \mbox{~if~} a x' \in \blank^* \mbox{~and~} x=\epsilon,\\ 
x \rw & \mbox{~if~} a x' \in \blank^* \mbox{~and~} x\neq \epsilon,
\end{cases}$
\item $f(x b\rw  x', a^{\rightarrow}) = 
\begin{cases}
x a c \rw x'' & \mbox{~if~} xa \notin \blank^*, x' = cx'', c \in \Gamma_0, \\ 
c \rw x'' & \mbox{~if~} xa \in \blank^*, x' = cx'', c \in \Gamma_0, \\ 
x a \blank \rw & \mbox{~if~} xa \notin \blank^*, x' = \epsilon, \\ 
\blank \rw  & \mbox{~if~} xa \in \blank^*, x' = \epsilon. 
\end{cases}$
\end{itemize}
A machine with one such store type is a Turing machine
with a one-way read-only input tape, and one read/write store tape. The store starts
off empty (a blank followed by the read/write head), and they can
extend in both directions as symbols are added to the
left and right. They can also shrink in size if everything to the right of the
read/write head is a blank, as with the left. This is exactly how
configurations of Turing machines change \cite{HU}.
Furthermore, $l$-reversal-bounded 1-tape Turing stores can be defined by restricting $L_{I}$ so that the number of alternations
between moving right and left on the tape is at most $l$.

For example, consider a deterministic machine $M = (Q,\Sigma,\Gamma,\delta,q_0,F)$ with a $2$-reversal-bounded Turing store accepting
$\{w \$ w \mid w \in \{a,b\}^+\}$, with $\Sigma = \{a,b\}, \Gamma = \{a,b,\blank,\rw\},  F = \{q_4\}$, and $\delta$ is as follows:
\begin{align*}
& \delta(q_0,c,\blank) = \{(q_1, c^{\rightarrow})\}, \forall c \in\{a,b\}, && \delta(q_1,c,\blank) = \{(q_1, c^{\rightarrow})\}, \forall c \in\{a,b\},\\
& \delta(q_1,\$,\blank) = \{(q_2,\blank^{\leftarrow})\} && \delta(q_2,\epsilon,c) = \{q_2,c^{\leftarrow})\}, \forall c \in \{a,b\},\\
& \delta(q_2,\epsilon,\blank) = \{(q_3,\blank^{\rightarrow})\}, && \delta(q_3,c,c) = \{(q_3,c^{\rightarrow})\}, \forall c \in \{a,b\},\\
& \delta(q_3,\lhd,\blank) = \{(q_4,\blank)\}.
\end{align*}
Despite $L(M)$ being non-context-free, the store language
$$
S(M) = \begin{array}[t]{l}\{q_4 x  \blank \rw\ \mid x \in \{a,b\}^+\} \cup \{q_3 x_1 \rw x_2 \mid \mbox{either~} x_1 \in \{a,b\}^+, x_2 \in \{a,b\}^* \mbox{~or~} x_1 \in \{a,b\}^+\blank, x_2 = \epsilon\} \cup\\
\{q_2 x_1 \rw x_2 \mid \mbox{~either~} x_1 \in \{a,b\}^+, x_2 \in \{a,b\}^* \mbox{~or~} x_1 = \blank, x_2 \in \{a,b\}^+\} \cup \{q_1 x_1 \blank \rw\ \mid x_1 \in \{a,b\}^+\} \\
 \cup \{q_0 \blank\rw\}, \end{array}
 $$
 which is a regular language.
\end{example}

%Let $\NCM(k,l)$ be the set
%of $l$-reversal-bounded $k$-counter machines,
%Let $\NCM$ be the set of reversal-bounded multicounter machines, and replacing
%D with N gives the deterministic variant (see e.g.\ \cite{Ibarra1978,InsertionOperationsJCSS}). 
Define machines with one pushdown
for the first tape, and $k$ additional reversal-bounded counters,
where each word in the store language is of the form 
$q x c_1^{j_1} \cdots c_k^{j_k}$, where $q$ is a state, $x$ is the contents
of the pushdown, and $j_1, \ldots, j_k$ are the contents of the counters.
Let $\NPCM$ be the set of machines with one pushdown, and 
some number $k$ of counters where
the pushdown is unrestricted, but the counters are reversal-bounded.
The family of languages accepted by $\NPCM$  
\cite{Ibarra1978,InsertionOperationsJCSS} is of interest since it has a
decidable emptiness and membership problem, and only accepts semilinear
languages.

Let $\NQA$ be the set of queue automata \cite{Harju2002278}.
%, see Appendix for definition)
%, see Appendix).
As with $\NPCM$, define machines with one queue for
the first tape, and $k$ additional counters. Let $\NQCM$ be the set
of machines with one queue, and some number $k$ of counters
where the counters are reversal-bounded (if the queue is also
reversal-bounded, these
only accept semilinear languages \cite{Harju2002278}, otherwise
they have the same power as Turing machines).
Let $\NSA$ be the set of stack automata \cite{StackAutomata,Stack2}.
Also, define machines with one stack for
the first tape, and $k$ additional counters. Let $\NSCM$ be the set
of machines with one stack, and some number $k$ of counters
where the counters are reversal-bounded (if the stack is reversal-bounded,
this implies that there is also a bound on the number of changes in direction of the read head when it reads
inside the stack structure).
Let $\NFPA$ be the set of $k$-flip pushdown machines, for some $k$ \cite{flipPushdown}. Replacing N with D gives each deterministic
variant.
%All of these models are one-way nondeterministic total, as they contain all the machines that can be constructed using
%the store types used (notice that $L_{I}$ can restrict the allowable sequences of instructions used e.g.\ for $k$-flip pushdown machines, but this is
%built into the store type rather than the machines; so one can still examine all machines that can be built using this store type).

%For a $k$-tape nondeterministic Turing machine ($\NTM$) with a read-only input tape (separate from the $k$ tapes) and $\Gamma$, the worktape
%alphabet, a word in the store language will be of the form
%$q x_1  \cdots  x_k$, where each $x_i$ is a representation of the contents of the $i$th tape (which we force to be over separate alphabets), of the form 
%$x \rw a x'$, $x,x' \in (\Gamma - \{\rw\})^*, a \in \Gamma - \{\rw\}$, where $x$ does not have any blank symbols (denoted
%by $\blank$) as prefix,
%and $x'$ does not have any blanks as suffix (as is
%the case in an instantaneous description of a $\NTM$ \cite{HU}). 
%As before, we force disjoint
%alphabets for each tape (excluding the read/write marker $\rw$ and the blank symbol $\blank$ 
%to leave no ambiguity as to what tape each symbol comes from.

\subsection{Store Languages of Turing Machines and Other One-Way Automata Models}
\label{onewayTMs}

Store languages have already been investigated for nondeterministic pushdown
automata.
It has been shown \cite{CFHandbook,GreibachCFStore} that the store language of each $\NPDA$ is a regular language. Moreover, the proof contains an effective construction. 
%Thus, $S(M)$ is a regular language as well.
\begin{proposition} \cite{CFHandbook,GreibachCFStore}
Given a one-way $\NPDA$ $M$, $S(M)$ is a regular language, and
$\SFam(\NPDA) \subseteq \LFam(\REG)$.
\label{NPDAStorage}
\end{proposition}
%The inclusions are trivially strict due to the states.

First, a general decidability proposition is proved for
machine models where the emptiness problem is undecidable.
\begin{proposition}
Let $\MFam$ be a set of machines defined using (potentially multiple)
store types, such that the emptiness problem
is undecidable for $M \in \MFam$. Then the problem, given $M \in \MFam$
and a word $x$,
determine whether $x \in S(M)$, is undecidable.
\end{proposition}
\begin{proof}
Let $M \in \MFam$ be a machine with initial state $q_0$ and 
initial store contents $z$ (which can be the concatenation of multiple
initial store contents for multi-store machines). 
Then $q_0 z$ is in the store language of $M$ if and only
if $L(M)$ is not empty. 

Hence, membership in $S(M)$ for $M \in \MFam$ is undecidable.
\qed
\end{proof}
This is true for sets of one-way machines, and also two-way machines
investigated later in the paper.
And in fact, it even holds for complexity classes, such as 
deterministic Turing machines with a one-way read-only input tape
and a logspace bounded worktape (the store).
These have a decidable membership problem but an undecidable emptiness
problem. Despite the languages accepted by these machines
being recursive (and in $\P$), membership in the store language
is undecidable (and so there cannot be an effective construction
to accept the store languages with another model, such as any
model with a decidable membership problem).

Next, store languages of restricted $\NTM$s will be studied.
They will be especially useful for characterizing store languages of other machine models.
In particular, $\NTM$s with a one-way read-only input tape
 and one reversal-bounded
read/write worktape are considered. In terms of languages accepted, these machines are powerful enough to simulate a number
of different machine models exactly, such as one-way nondeterministic reversal-bounded pushdown automata, reversal-bounded
queue automata, reversal-bounded stack automata, and reversal-bounded $k$-flip pushdown automata, where the worktape
acts exactly like the other stores. 
%as follows:
%\begin{proposition}
%\label{languageSimulation}
%Given $M$, a reversal-bounded pushdown automaton, then $L(M)$ can be accepted by a nondeterministic Turing machine with
%a one-way read-only input, and a reversal-bounded read/write worktape.
%This result also holds for $M$, a reversal-bounded queue machine,
%a reversal-bounded stack machine, and a reversal-bounded $k$-flip pushdown automaton.
%\end{proposition}
%Indeed, a machine from each of the respective classes can be simulated by a nondeterministic Turing machine where the worktape
%acts exactly like the storage from the original model, and since the direction that the storage is changing is the same between
%each reversal, the Turing machine read/write head only needs to make a bounded number of changes in direction. 

Next, the store languages of these Turing machines are examined. Although $\NTM$s in
general have non-regular store languages (investigated in Section \ref{sec:bounds}), when there is only one worktape, and it is reversal-bounded,
the store languages are always regular.
\begin{proposition} 
\label{TMReg}
Let $M= (Q,\Sigma,\Gamma,\delta, q_0,F)$ be an $\NTM$ with a one-way read-only input tape and a reversal-bounded
read/write worktape.  Then $S(M) \in \LFam(\REG)$.
%Further, for all states $q$ of $M$, $\acc(q)$ and $\coacc(q)$ are in $\LFam(\REG)$.
\end{proposition}
\begin{proof} 
%Clearly, we may assume that on stationary moves, i.e., when the
%read/write head on the worktape enters a cell from the right or from the left, it can only rewrite the cell at most once.  (This is because if $M$ makes a
%stationary move on the cell, $M$ can simulate the symbol changes 
%on the cell in its finite-control and when it it leaves the cell, can rewrite
%the cell.)
Let $M$ make at most $l$ reversals on the worktape.  
Note that $L(M) \subseteq \Sigma^*$, and $S(M) \subseteq Q \Gamma^*$.
Let $\Gamma_0 = \Gamma - \{\rw\}$, and $\Gamma_0' = \{a' \mid a \in \Gamma_0\}$, a new alphabet (each letter is a ``primed''
version of a letter in $\Gamma_0$, including a primed version of the blank symbol $\blank$). 
Define a new alphabet $C$ whose symbols have ``tracks'', with less than or equal
to $(l+2)$-tracks 
of the form $(a_1,  a_2, \ldots, a_p), p \leq l+2$  where
$a_1$ is in $\Gamma_0 \cup \Gamma_0' \cup Q$ and $a_i \in \Gamma_0 \cup \Gamma_0'$, for each $i$, $2 \leq i \leq p$. 

An intermediate 2-way $\NFA$ $M'$ is constructed,
whose input is in $C^*$  delimited 
with end-markers $\rhd$ and $\lhd$. 
Thus the input to $M'$ is $\rhd w \lhd$, where 
$w \in C^*$.   The input $w$ can be thought of as 
having less than or equal to $l+2$ tracks.

Intuitively, $M'$ is trying to verify that the contents of the first track represents a configuration in
an accepting computation of $M$, where a symbol $a' \in \Gamma_0'$ is
used in place of $a\rw $ in the store language. To do this, $M'$ nondeterministically guesses an input $x\in \Sigma^*$
and simulates $M$ on $x$, track $2$ is verified to be the initial
store contents, each track from tracks $3$ to $p-1$ is verified to be
the store contents at a point of reversal, and track $p$ is verified
to be a final accepting configuration. All tracks are padded by
blank symbols to all be of the same length.

 Then 
$M'$ operates as follows:
\begin{enumerate}
\item The first track is verified to contain a word $\blank^n q w' \blank^l$, where
$q \in Q, n,l \geq 0$,
$w'  \in \Gamma_0^* \Gamma_0' \Gamma_0^*$ 
that does not start or end with $\blank$ 
(below, $M'$ will verify that the word obtained from $qw'$ by replacing
$a'$ with $a \rw$ is in $S(M)$).

\item $M'$ goes to the left end-marker $\rhd$.  $M'$ checks that the 
second track contains $\blank^m \blank' \blank^r $  for some $m,r \geq 0$ (the worktape
starts off with only blanks, 
it is implied that all tracks are of length $m + r + 1 = n+ l + |w'| +1$).

\item \label{simulationstep} $M'$ then simulates the $\NTM$ $M$ on a guessed input $x\in \Sigma^*$, letter-by-letter, but instead of writing, verifies
the next track contents is an updated version of the current
track at the next point of reversal.
Between the initial configuration and the first reversal, between every two reversals, and between the last reversal and the final configuration,
$M'$ checks that the contents
of  track $i+1$ is the updated contents of the worktape from the worktape
on track $i$ for $i = 2, \ldots, p-1$. 
To do this, if track $i$ contains
$\blank^{\alpha} x a' y \blank^{\beta}, x,y \in \Gamma_0^*, a \in \Gamma_0'$,
and say $i$ is even (the case is similar if $i$ is odd), then it is verified that track $i+1$ starts with
$\blank^{\alpha}x$. Then, if $M$ uses a transition that replaces $a$
with $b$ and moves right on the worktape, then track $i+1$ has $b$ next (a sequence
of transitions that stays on the same storage cell are remembered
in the finite control), and this simulation continues until
$M$ makes a reversal. When 
$M$ makes a reversal, say from moving right to left, $M'$  first
``marks'' the point of reversal by reading a primed symbol
in that position of track $i+1$ (storing the read/write head in track $i+1$), 
then it moves to the right end-marker and checks that each symbol in track
$i$ from the point of reversal to the right end-marker matches the
 symbols in track $i+1$ (this also implies that track $i+1$
 has exactly one symbol from $\Gamma_0'$).  $M'$  then moves left back to the
point of reversal (which is retrievable from the primed symbol), 
and resumes the simulation using the next track from the current track.

\item At some nondeterministically guessed reversal of the simulation
as described in step \ref{simulationstep} (say
while scanning track $i$, and track $i+1$ reverses from left to right),
while $M'$ is verifying that track $i+1$ follows from track $i$, in parallel,
$M'$ verifies that the contents of track 1 is a configuration
of the Turing machine between these two reversals. 
To do this, $M'$ remembers the state $q$ on track 1, then compares track $1$ to track $i+1$ symbol-by-symbol, 
until reaching the read/write head in track 1, where the current
state of the simulated machine is verified to be $q$ and
the remaining part of track $1$ is verified to be the same as track $i$. Thus, track
$1$ is a configuration between tracks $i$ and $i+1$. $M'$ then continues
the simulation as in step 3.
\end{enumerate}

Since $M'$ can only read and not write (on the input), it just verifies
the moves and that the changes in the symbols of the $\NTM$ $M$ on track $i$ 
are reflected in the $i+1$st track.
$M'$ accepts if $M$ accepts and step 4 above has been successful.

It is known that 2-way $\NFA$s accept only regular languages
\cite{HU}.
Then, apply a gsm \cite{HU} to extract just the word $w'$ from the first track, erasing blanks appropriately,
and replacing any symbol $a' \in \Gamma'$ by $a\rw $.
Since regular languages are closed under gsm mappings, the result follows.
\qed \end{proof}

As $\NTM$s with a reversal-bounded worktape only give regular store languages, generalizations of
these $\NTM$s that still have a decidable emptiness problem are also of interest.
In \cite{Harju2002278}, it was shown that such $\NTM$s augmented by reversal-bounded
counters have a decidable emptiness problem. Therefore, understanding the store
languages of this model is valuable, which is studied next.
The proof
uses $2\NCM$s, which are two-way nondeterministic reversal-bounded multicounter
machines, together with a similar technique as in the proof of Proposition
\ref{TMReg}, ultimately determining that the store languages
are accepted by one-way $\NCM$ machines.

\begin{proposition}
\label{TMCounters}
Let $M$ be an $\NTM$ with a one-way read-only input tape, a reversal-bounded
read/write worktape, and $k$ reversal-bounded counters. Then $S(M) \in \LFam(\NCM)$.
\end{proposition}
\begin{proof} 
Here, the store consists
of the state, read/write tape, and the values of the counters.
The proof of Proposition \ref{TMReg} is generalized (using the
same alphabets).

Construct an intermediate $2\NCM$ $M'$ that is reversal-bounded on the input tape, with $2k$ reversal-bounded counters.
% which operates as follows:

\begin{enumerate}
\item $M'$ will have as input $z =  w c_1^{i_1} \cdots c_k^{i_k}$  
(with end-markers)
where $w$ has multiple read-only tracks (over the alphabet $C$, just like in Proposition \ref{TMReg}), and the first track is $\blank^n q w' \blank^l,
q\in Q, n,l  \geq 0, w' \in \Gamma_0^* \Gamma_0' \Gamma_0^*$
that does not start or end with $\blank$.

\item $M'$ simulates $M$'s
reversal-bounded read/write worktape on the tracks of the read-only $w$
(as in Proposition \ref{TMReg}) and using reversal-bounded counters to simulate the
reversal-bounded counters of $M$ faithfully. However,
$M'$ keeps two copies of each counter, where the two sets of counters
are updated synchronously (and are therefore identical during the first part of the simulation). 

\item At some point, $M'$ nondeterministically guesses that the contents 
of track 1, $qw'$ together with counter values $(i_1, \ldots, i_k)$ encoded in the input is a representation of
a configuration between the current track and the next track. Then, as in Proposition \ref{TMReg},
$M'$ matches the symbols in the first track with track $i+1$ until the
symbol from $\Gamma_0'$ in track 1, and
if so, stops updating one set of the counters.
Then $M'$ continues by matching track $1$ with track $i$.  Then the
simulation of $M$ continues, using the other set of counters and on the
remaining tracks (as in Proposition \ref{TMReg}). At the end of the simulation, $M'$ verifies
that the non-updated set of counters has the same values as $i_1,\ldots, i_k$ encoded on the input $c_1^{i_1} \cdots  c_k^{i_k}$. If this is the case, then
$qw' c_1^{i_1} \cdots c_k^{i_k}$ is indeed an intermediate store
configuration of an accepting computation, 
and if so, $M'$ accepts.
\end{enumerate}

Now $M'$ is a reversal-bounded (on the input) $2\NCM$. Hence, $M'$ can be
converted to an equivalent one-way $\NCM$ $M''$ , i.e., $L(M'') = L(M')$
(this can be done even for the more general finite-crossing $2\NCM$s)
\cite{Gurari1981220}.
%(This uses my past result
%with Gurari.)

Since $\LFam(\NCM)$ is closed under gsm mappings (follows from 
closure under homomorphism, inverse homomorphism, and intersection with regular
languages \cite{Ibarra1978}), 
construct an $\NCM$  $M'''$ that applies a gsm that extracts 
$qw' c_1^{i_1} \cdots c_k^{i_k}$ (and replaces $a' \in \Gamma'$
with $a \rw $)
from the first track of $w$ and $z$.
It follows that $S(M)$ is in $\LFam(\NCM)$.
\qed \end{proof}

It is possible to accept
the store languages of other machine models,
such as reversal-bounded $\NQCM$ with $\LFam(\NCM)$
(this does not follow
directly
from the fact that such Turing machines can simulate the input languages of this model, as store
languages rather than input languages are of interest here). 
\begin{proposition}
\label{TMQueue}
If $M$ be a reversal-bounded $\NQCM$, then $S(M) \in \LFam(\NCM)$. 
\end{proposition}
\begin{proof}
Given $M$ with $k$ counters, construct an intermediate $\NTM$ $Z$ with an input tape plus
one reversal-bounded read/write worktape, and
$k$ additional reversal-bounded counters, whose resulting store
language will be in $\LFam(\NCM)$ by Proposition \ref{TMCounters}.
%On input,
%$y=w c_1^{i_1} \cdots c_k^{i_k}$ (alphabet of $w$ disjoint from
%$\{c_1, \ldots, c_k\}$), $Z$ accepts if and only if $y \in S(M)$.   
$Z$ operates as follows:  
every time $M$ enqueues $y = b_1 \cdots b_m$, with $b_i$ being a letter, $m \geq 1$, $Z$ writes $y$ to the right end of the read/write worktape writing one letter at a time (using new intermediate states) and
simulating the counters exactly. Every time $M$ dequeues, $Z$ writes blank characters to
the left end of the tape towards the right (thus removing characters from the store
of the Turing machine as well). Each time such a reversal occurs (switching between enqueueing and
dequeueing), the Turing machine moves
its tape head from one end of the tape to the other on new intermediate states. Since the queue
is reversal-bounded, the Turing tape is reversal-bounded as well.

The store language of $Z$ has each word of the form 
$qw c_1^{i_1} \cdots c_k^{i_k}$ and is in $\LFam(\NCM)$ by
Proposition \ref{TMCounters}. From there, $S(Z)$ will be transformed
into $S(M)$ by a gsm $g$. Indeed,
the read/write head does not appear explicitly in the store language of the queue nor the blank symbol before the read/write head, whereas it does in the
Turing machine, but they can be removed by $g$. 
%Also, 
%the simulated state $q'$ of $M$ is not necessarily the same as the state $q$ of $Z$, but $q'$
%can be determined by $q$. 
Furthermore, if $M$ enqueues more
than one symbol ($m>1$), then $Z$ requires $m$ moves. Then all intermediate
states used by $Z$ when writing each $b_i, i<m$ are not mapped by $g$ as they do not have
corresponding configurations of $M$. Similarly, the
intermediate states that $Z$ uses when it switches between simulating enqueuing and dequeuing instructions (by moving the tape head to the opposite end) are not mapped by $g$. Since $\LFam(\NCM)$ is closed under
gsm mappings, the store language of $M$ is in $\LFam(\NCM)$.
\qed \end{proof}
Note that in the proof of the result above, if there are no counters in $M$, then only a Turing machine with one
reversal-bounded read/write worktape is required, whose store language is a regular language
by Proposition \ref{TMReg}.
\begin{corollary}
If $M$ is a one-way reversal-bounded queue automaton, then
$S(M) \in \LFam(\REG)$. 
\end{corollary}

Next, it will be shown that the same is true for reversal-bounded
stack automata augmented by reversal-bounded counters. Recall that stack automata can operate like pushdown automata with additional instructions that can read in the pushdown store in
a read-only fashion \cite{Stack2}.

\begin{proposition}
Let $M$ be a reversal-bounded $\NSCM$ or a reversal-bounded $\NPCM$.
Then $S(M) \in \LFam(\NCM)$.
\end{proposition}
\begin{proof}
Let $M$ be a reversal-bounded stack automaton with $k$ counters. 
Construct a Turing machine $M'$ 
with a one-way input tape, a reversal-bounded worktape, and $k$ counters.
A stack automaton
is very similar to a restricted type of Turing machine with a one-way input tape and a worktape to simulate
the stack that only changes values at the right end of the tape, except for the following minor differences: 
instructions that read from the inside of the stack are simulated by transitions that move but do not change from the inside of the Turing tape, the bottom-of-stack marker can be initially placed on the tape, the top of the stack marker is simulated with a blank, and a stack automaton allows to push multiple symbols in one transition. The latter type
can be simulated with new intermediate
states  that push one symbol at a time.
The counters
are simulated faithfully.
Then $S(M') \in \LFam(\NCM)$, by Proposition \ref{TMCounters}.

 As the intermediate configurations of $M'$ that are involved in simulating
 push transitions of more than one symbol are not configurations
of $M$,  a
gsm can be used to not output on those intermediate configurations
(similar to the proof of Proposition \ref{TMQueue}).
From the differences described, it is clear that
$S(M)$ is in $\LFam(\NCM)$. The proof is similar for $\NPCM$s.
\qed \end{proof}
This immediately implies that every reversal-bounded stack automaton
has a regular store language, but this result will be improved
later in the paper.

Next, a $k$-flip-pushdown automaton is a pushdown automaton, with the ability to flip its store. This can happen at most $k$ times in an 
accepting computation (see Example \ref{flipdefinition}). Despite the additional ability to flip the
store, regularity of the store language is preserved.
\begin{proposition}
If $M = (Q,\Sigma,\Gamma,\delta,q_0,F)$, $k \geq 0$
is a one-way $k$-flip pushdown automaton,
then $S(M) \in \LFam(\REG)$.
\end{proposition}
\begin{proof}
For an $\NPDA$ $M'$ with state $q$,
let $\acc_{M'}(q) = \{q x \mid (q_0, w, Z_0) \vdash_{M'}^* (q,\epsilon, x) \}$
and $\coacc_{M'}(q) = \{qx \mid (q,w,x) \vdash_{M'}^* (q_f,\epsilon, x'), q_f \in F \}$. 
That is, $\acc_{M'}(q)$ is the set of store contents in state $q$ that are reachable from the initial configuration, and $\coacc_{M'}(q)$ is the set of store contents in state $q$
that can reach an accepting configuration. 
It is known that for $\NPDA$s $M'$ and all states $q$, both $\acc_{M'}(q)$ and $\coacc_{M'}(q)$
are regular \cite{CFHandbook}. 

Every instruction of a $k$-flip $\NPDA$ is either a standard $\NPDA$ instruction or a flip instruction, as defined in Example \ref{flipdefinition}, and at most $k$ flip instructions can be applied in every accepting computation. First, note that the store language $S(M)$ is the union of the
store languages obtained using each final state separately. Since the regular languages
are closed under union, assume without loss of generality that $M$ only has one
final state $q_f$.

Consider any computation of $M$, not necessarily from an initial configuration nor to an accepting
configuration,
$$\alpha: (r_0,w_0,\gamma_0 ) \vdash \cdots \vdash (r_m , w_m, \gamma_m),$$
$r_j \in Q, w_j \in \Sigma^*, \gamma_j \in Z_0 \Gamma_0^*$, for $ 0 \leq j \leq m$ using
at most $k$ flips. From $\alpha$, there is a sequence, denoted by 
\begin{equation} f(\alpha) =  p_1, \ldots, p_{2l} \in X,
\label{flipequation}
\end{equation}
such that $\alpha$ has $l \leq k$ flip transitions, and the $i$'th
flip transition applied is from $p_{2i-1}$ to $p_{2i}$, for all $1 \leq i \leq l$.
Also, if $\alpha$ is an accepting computation (starting from an initial configuration and ending in an accepting configuration), let
$S_{\alpha}(M) = \{r_0\gamma_0, \ldots, r_m \gamma_m\}$. This is generalized to sets of accepting
derivations $Y$, as $S_Y(M)$.

Let $X$ be the finite set of all sequences $z = p_1, \ldots, p_{2l}$, where $l \leq k$, and there is some
transition of $M$ that flips while switching from $p_{2i-1}$ to $p_{2i}$, for all $i$. Given any $z \in X$,
let $g(z)$ be the set of all accepting computations $\alpha$ of $M$ such that $f(\alpha) = z$. It is clear
that $S(M) = \bigcup_{z \in X} S_{g(z)}(M)$. Thus, it is enough to show that, for each $z \in X, S_{g(z)}(M)$ is regular. 

Let $z = p_1,\ldots, p_{2l} \in X$. Then each word in $S_{g(z)}(M)$ is either derived from a
configuration between the $j$'th flip and the $(j+1)$'st flip, for $0 \leq j < l$, or after the
$l$'th flip (and before the end). For $0 \leq j \leq l$, let $S_{g(z),j}$ be all those store
contents between the $j$'th flip (or the beginning if $j=0$) and the transition before the
$j+1$'st flip (or the end of the computation if $j=l$). Again, if each
$S_{g(z),j}(M)$ is regular, then $S_{g(z)}(M)$ is regular. Let $0 \leq j \leq  l$. For each $i$ from
$0$ to $j$, let $z_i = p_1,\ldots, p_{2i}$ (if $i=0$, then there are no flip transitions), and
let $\acc_i(q) = \{qx \mid \alpha: (q_0,w,Z_0 ) \vdash^* (q,\epsilon, x), \alpha \in g(z_i)\}$. 
It will be shown that each 
$\acc_i(q)$ is regular. This will be done by building an $\NPDA$ $M'$ without flips such that
$\acc_{M'}(q) = \acc_i(q)$, which then must be regular. This is done inductively on $z_i$.

Consider the $\NPDA$ $M_0$ obtained from $M$ by keeping all $\NPDA$ transitions but omitting flip transitions. 
Then $\acc_{M_0}(p_1) = \acc_0(p_1)$ is regular. Hence, $Y_1 = p_2 Z_0( (p_1Z_0 )^{-1}\acc_0(p_1))^R$
is regular as well (since the regular languages are closed under left quotient, concatenation, and reversal.
This is exactly the set of store contents that can be derived from those in $\acc_0(p_1)$
via a flip transition from $p_1$ to $p_2$.
Then build another $\NPDA$ $M_1$ that pushes an arbitrary word of $Y_1$, then
simulates $M$ without flips until $p_3$. Then $\acc_{M_1}(p_3) = \acc_1(p_3)$, which is
again regular. This same procedure proceeds inductively by reversing the regular store language until an 
$\NPDA$ $M_j$ can be built such that for each $p \in Q$, $\acc_{M_j}(p) = \acc_j(p)$, which is regular.

Similarly, for each $i$ from $j$ to $l$, let $z_i = p_{2i+1}, \ldots, p_{2l}$, and
$\coacc_i(q) = \{qx \mid \alpha: (q,w,\gamma) \vdash^* (q_f,\epsilon, \gamma'), \alpha \in g(z_i)\}$.
In a similar fashion, for each $p \in Q, \coacc_{M_j}(p) = \coacc_j(p)$, which is again regular.
Furthermore, $\bigcup_{p \in Q}(\acc_j(p) \cap \coacc_j(p)) = S_{g(z),j}(M)$. Hence,
$S_{g(z),j}(M)$ is regular, $S_{g(z)}(M)$ is regular, and $S(M)$ is regular.
\qed \end{proof}
This is indeed quite a general family to have only regular store languages.

Next, the store languages of $\NCM$s are analyzed. Surprisingly,
only deterministic machines in $\DCM$ are needed to accept them.
\begin{proposition}\label{NCMtoDCM}
If $M$ is an $\NCM$,  then $S(M) \in \LFam(\DCM)$.
Thus, $\SFam(\NCM) \subsetneq \LFam(\DCM)$.
\end{proposition}
\begin{proof}
Let $M$ have counters $C_1, \ldots, C_k$.  First, construct an intermediate $\NCM$ $M'$
with counters named $C_1, \ldots, C_k, D_1, \ldots, D_k$ to accept $S(M)$. 
Then  $M'$,  when given an input  $z$,
checks that $z$ is of the form  $q c_1^{i_1} \cdots c_k^{i_k}$ for 
a state $q$, $i_1, \ldots, i_k \geq 0$
(this can be done by a $\DFA$ in parallel).  To check that $z$ is in $S(M)$,  
on transitions that do not read any input, $M'$ guesses an  input $x$ to $M$ in a letter-by-letter fashion and simulates $M$ on $x$ using counters $C_1, \ldots, C_k$ and 
$D_1, \ldots, D_k$ (i.e., $D_1, \ldots, D_k$ are duplicate counters which  operate
like $C_1, \ldots, C_k$ similar to Proposition \ref{TMCounters} step 2).  At some point (nondeterministically chosen), $M'$
stops  updating counters $D_1, \ldots, D_k$ and remembers the current state $q'$ but  continues the simulation
with counters  $C_1, \ldots, C_k$.  When $M$ accepts, $M'$ checks that the
value in counter $D_j$ is $i_j$ (from the input), for all $j$, $1 \leq j \leq k$ and that $q = q'$. Hence, $S(M) = L(M')$. The $\NCM$ $M'$ can then be converted
to a $\DCM$ $M''$, since it is known that any $\NCM$ accepting a bounded language 
can be accepted by a $\DCM$ \cite{IbarraSeki}.
\qed \end{proof}

Although all store languages of the nondeterministic model $\NCM$ can be accepted by
the deterministic model $\DCM$, next, it will be shown that this is not the case for $\DPCM$.
\begin{proposition}
Let $M$ be an $\NCM$ which accepts with all counters
zero and in a unique accepting state $f$ which is
never re-entered.  Let $L = fZ_0 L(M)$.
Then there is a 0-reversal-bounded
$\DPCM$ $M'$ and a regular
language $R$ such that $L = S(M') \cap R$. Then $L(M) = (fZ_0)^{-1} S(M')$.
\label{cannotremovenondet}
\end{proposition}
\begin{proof}
Let $M$ be an $\NCM$ with $k$ counters and input alphabet $\Sigma$.
Represent each transition of $M$ by an abstract symbol:
$$[(q, a, d_1, \ldots, d_k) \rightarrow  (p, y_1, \ldots , y_k)],$$
where $p$ and $q$ are states, $a$ is either in $\Sigma$ or $\epsilon$,
$d_i$ represents the status (zero or non-zero)
of counter $i$, and  $y_i$ is the change in counter $i$.
Let $\Delta$ be the set of symbols representing the transitions.

The input alphabet of the $\DPCM$ $M'$ is $\Delta$.  For a string
$y \in \Delta^*$, let $x$ be the concatenation of the the input
components of the transitions in $y$.  On input $y$, $M'$
writes $x$ on the stack while simulating the computation of $M$ on
$x$ using the counters, and accepts in state $f$ if $M$ accepts $x$. $M'$ is indeed
deterministic since each symbol of $\Delta$ implies the transition
to apply.

Hence, $S(M')$ contains all strings of the form $fZ_0w$,
where $w$ is in $L(M)$.  Let $R$ be the regular language $f Z_0\Sigma^*$.
Then $L = S(M') \cap R$, and $L(M) = (fZ_0)^{-1}S(M')$.
\qed \end{proof}

From this, the following is true:
\begin{proposition}
\label{0revbounded}
There is  a 0-reversal-bounded $\DPCM$ $M$ such that $S(M)$
cannot be accepted by any $\DPCM$.
\end{proposition}
\begin{proof}
Suppose otherwise.
It is known that there are languages in $\LFam(\NCM)$ that are
not in $\LFam(\DPCM)$ \cite{OscarNCMA2014journal}. 
Let $L$ be such a language accepted by some $M \in \NCM$. By Proposition \ref{cannotremovenondet},
there exists $M' \in \DPCM$ that is $0$-reversal-bounded
such that $(fZ_0)^{-1} S(M') = L(M)$. But $S(M') \in \LFam(\DPCM)$
by the assumption. Also, it is clear that $\LFam(\DPCM)$ is closed
under left quotient with a fixed word as a $\DPCM$ can simulate first on the fixed word deterministically, then on the input deterministically. Hence, $L(M) \in \LFam(\DPCM)$,
a contradiction.
\qed \end{proof}

Lastly, two results will be stated that
are shown below in Section \ref{subsec:twoway} and Section \ref{subsec:twowayoneway} respectively, which
are results on one-way stack automata, but require results on two-way
automata for their proofs. The first is already known \cite{KutribCIAA2016}.
\begin{proposition} \cite{KutribCIAA2016}
If $M$ is a one-way nondeterministic stack automaton, then $S(M) \in \LFam(\REG)$.
\end{proposition}

\begin{proposition}
There exists $M \in \NSCM$ with one $1$-reversal-bounded counter over a
unary alphabet such that $S(M) \notin \LFam(\NPCM)$.
\end{proposition}

The latter result is interesting in the following
sense: 
an $\NSCM$ combines a stack and reversal-bounded counters.
A stack alone yields only regular store languages;
but the store languages of $\NSCM$ are
more general, describing some languages that are neither in $\LFam(\NCM)$ nor 
$\LFam(\NPCM)$.
However, it is seen next that $\NSCM$s only
yield $\NSCM$ store languages.
\begin{proposition}
If $M \in \NSCM$, then $S(M) \in \LFam(\NSCM)$.
\end{proposition}
\begin{proof}
Let $M$ be a $k$-counter $\NSCM$. Construct a 
$2k+2$ counter $\NSCM$ machine $M'$ (give names $c_i, d_i, e,f$ to the counters, $1 \leq i \leq k$) to accept $S(M)$ that simulates $M$ with two
identical copies of each counter, $c_i$ and $d_i$, $1 \leq i \leq k$, on a guessed input. Then,
at some nondeterministically guessed spot, $M'$ verifies that the stack contents are
the same as the input by moving the stack head to the left-end-marker while adding one to counter $e$ and $f$, then $M'$ verifies that the stack contents are the same as the input by comparing the stack to the input symbol-by-symbol while
decreasing $e$ to verify that the read/write head on the input is in the correct location. Then $M'$ returns its stack read head to the proper location using counter $f$. Then $M'$ verifies that the counter values match the input values by
decreasing each $c_i$. Finally, $M'$ continues the simulation on the
second set of counters, $d_i$, accepting if $M$ accepts.
\qed \end{proof}

\subsection{Connections Between Deterministic and Nondeterministic Machines}

Thus far, the primary concern has been store languages
of nondeterministic machine models.
In this section, a connection between deterministic and
nondeterministic one-way machines is demonstrated. 
A store type is said to have stay instructions if there are instructions to keep the store the same (that do not violate the instruction language). All store types considered in this paper are of this form.
\begin{proposition}
\label{detandnondet}
Let $\Omega_1, \ldots, \Omega_k$ be store types with stay instructions, and let $M$ be a one-way nondeterministic $(\Omega_1, \ldots, \Omega_k)$-machine.
One can construct a one-way deterministic $(\Omega_1, \ldots, \Omega_k)$-machine
acceptor $M'$  of the same type as $M$
such that $S(M') = S(M)$. 
\end{proposition}
\begin{proof}
Let $t_1, \ldots, t_m$ be new symbols in bijective correspondence
with the transitions of $M$. 
Then let $M'$ operate as follows over the input alphabet
$T = \{t_1, \ldots, t_m\}$, with the same state set, initial state, and final state set,
and the transition function $\delta'$ is as follows: for each transition of $M$,
$t_i: (p,\iota_1, \ldots, \iota_k) \in \delta(q,a,d_1,\ldots, d_k),a \in \Sigma \cup \{\epsilon\}$, create a transition
of $M'$ $(p,\iota_1,\ldots, \iota_k) \in \delta'(q,t_i,d_1,\ldots, d_k)$. Also, create transitions that stay from 
any final state while reading the end-marker. Thus, consider
any accepting computation of $M$ using a sequence of transitions. Then reading the
corresponding sequence labels with $M'$ (followed by reading the end-marker) is an accepting computation with the
store changing identically (thus, it is in the instruction language). Similarly, given any accepting computation of $M'$, applying this
sequence accepted by $M'$ as a sequence of transitions of $M$ is accepting with the store changing identically. Hence,
$S(M') = S(M)$. Also, $M'$ is deterministic since the input symbol dictates the transition
to apply.
\qed \end{proof}
Also, note in Proposition \ref{detandnondet} that $M'$ operates in realtime.

\begin{corollary}
The following are true: $\SFam(\NPDA) = \SFam(\DPDA)$ and $\SFam(\NCM) = \SFam(\DCM)$
\end{corollary}
This is also true for all one-way nondeterministic and deterministic machine
models considered in this paper.

Finally, it is interesting to consider whether store languages of machine models
can always be accepted by only deterministic machines of the same type. Indeed, the 
following are true:
\begin{enumerate}
\item
If $M$ is an $\NFA$ or an $\NPDA$, then $S(M)$ is regular and
hence the deterministic version of the model can accept
its own store language.
\item
If $M$ is an $\NCM$, then $S(M)$ can be accepted by a $\DCM$,
hence the deterministic version of the model can accept
its own store language.
\end{enumerate}
However, it follows from Proposition \ref{0revbounded} that the store languages of $\DPCM$s
cannot be accepted by $\DPCM$s.

\section{Machines with Two-Way Read-Only Inputs}
\label{sec:twoway}

Using exactly the same store types as defined in the previous section, two-way input machines can also be
defined. 
Given store types $\Omega_1, \ldots, \Omega_k$ with
$\Omega_i = (\Gamma_i, I_i, f_i, g_i, c_{0,i}, L_{I,i}), 1 \leq i \leq k$,
two-way inputs have an end-marker on both sides,
$\rhd w \lhd$, and the finite transition relation is from
$Q \times (\Sigma \cup \{\rhd,\lhd\}) \times \Gamma_1 \times \cdots \times \Gamma_k$
to $Q \times  I_1 \times \cdots \times I_k \times \{-1,0,+1\}$ (with the last component describing the direction of the input
head movement), a configuration of $M$ is a tuple
$(q,\rhd w \lhd, \gamma_1, \ldots, \gamma_k, j)$, where $q \in Q, w\in \Sigma^*,\gamma_i \in 
\Gamma_i^*, 1 \leq j \leq |w|+3$ giving the current position
on the input (position $1$ is $\rhd$, $|w|+2$ is $\lhd$, and $|w|+3$ is off the input tape). The
derivation relation $\vdash_M$ is defined by:
$(q,\rhd w\lhd,\gamma_1, \ldots, \gamma_k, j) \vdash_M (q', \rhd w \lhd, \gamma_1', \ldots, \gamma_k', j')$
if there exists $(q', \iota_1, \ldots, \iota_k, n) \in \delta(q,a,d_1, \ldots, d_k)$, $a$ is the $j$'th
character of $\rhd w \lhd$, $j' = j+n, g(\gamma_i)=d_i, f(\gamma_i,\iota_i) = \gamma_i'$, for each $i$, $1 \leq i \leq k$. Validity is defined just like with one-way
machines.
The language accepted by $M$, $L(M) = \{ w \mid (q_0, \rhd w \lhd, c_{0,1}, \ldots, c_{0,k},1)
\vdash_M^x (q_f, \rhd w \lhd, \gamma_1, \ldots, \gamma_k, j), q_f \in F, w \in \Sigma^*,
\gamma_i \in \Gamma_i^*, 1 \leq j \leq |w|+3, x \mbox{~is valid}\}$.
The store language of $M$, $S(M)$ is equal to
$\{q \gamma_1 \cdots \gamma_k \mid (q_0, \rhd w \lhd, c_{0,1}, \ldots, c_{0,k}, 1) \vdash_M^x 
(q,\rhd w \lhd, \gamma_1, \ldots, \gamma_k, j) \vdash_M^y (q_f, \rhd w \lhd, \gamma_1', \ldots, \gamma_k', j'),
q \in Q, q_f \in F, \gamma_i, \gamma_i' \in \Gamma_i^*, j,j' \in \{1, \ldots,
|w|+3\}, xy \mbox{~is valid} \}$.

In the previous section, store languages of
different types of machines with a one-way read-only input were studied.
The rest of this section will investigate store languages of machine models with two-way inputs.

\subsection{Two-Way $\NCM$s and $\DCM$s}
\label{subsec:twoway}

This subsection considers store languages of two-way $\NCM$s ($2\NCM$s)
and two-way $\DCM$s (2$\DCM$s).
A machine is {\em finite-crossing} if
there is a $d\in \mathbb{N}$ such that in any computation, 
the input head crosses the boundary between any two 
adjacent cells of the input no more than $d$ times.
The first result demonstrates the surprising fact that
the store languages of finite-crossing $2\NCM$s can always be
accepted by machines that are only one-way and deterministic.
\begin{proposition} If $M$ is a finite-crossing $2\NCM$, then $S(M)
\in \LFam(\DCM)$.
\end{proposition}
\begin{proof}
Given a $k$-counter finite-crossing $2\NCM$ $M$ over $\Sigma$, first, construct an intermediate finite-crossing $2\NCM$ $M_1$ with $2k+1$ counters and input of the form
$x q c_1^{i_1}  \cdots c_k^{i_k}$, where $x \in \Sigma^*$.

$M_1$ simulates the computation of $M$ on $x$ with two sets of counters
named $C_j$ and $D_j$ for $1 \leq j \leq k$ so that $C_j$ and $D_j$ contain identical values.  At some nondeterministically
chosen point, $M_1$ stores the input head position in the remaining counter, checks that the input segment $q c_1^{i_1} \cdots c_k^{i_k}$
corresponds to the simulated state and the value $i_j$ is equal to the value stored in counter $D_j$, for each $1 \leq j \leq k$.  If so,
$M_1$ continues the simulation of $M$ on the correct position of $x$ (which can be recovered)
using the $C_1, \ldots, C_k$ counters, and accepts if and only
if $M$ accepts.
It is known that all finite-crossing $2\NCM$s can be converted to
one-way $\NCM$s \cite{Gurari1981220}, and so convert $M_1$ to an $\NCM$ $M_2$ and then construct an $\NCM$
$M_3$ that erases the $x$ \cite{Ibarra1978} ($\LFam(\NCM)$ is closed under homomorphisms and therefore $x$ can be erased).  Then convert $M_3$ to a $\DCM$ $M_4$ as all bounded $\NCM$ languages are in $\LFam(\DCM)$ \cite{IbarraSeki}.
\qed \end{proof}

Without the finite-crossing condition, this no longer holds.
\begin{proposition} 
\label{finiteCrossing2DFA}
There is a non-finite-crossing $2\DCM$ $M$ with
one $1$-reversal counter over the bounded language $a^* b^*$ such that
$S(M) \notin \LFam(\NPCM)$.
\end{proposition}
\begin{proof}
Construct $2\DCM$ $M$ with one $1$-reversal counter
accepting $\{a^ib^j \mid  i, j > 1, i \neq j,  i \mbox{~is a multiple of~} j\}$
as follows: $M$ stores $i$ in counter $C$ and enters a distinguished state $f$.
(Thus, the configuration at this time is $fc_1^i$.) Then $M$ changes
state and checks (by decrementing $C$ while going back-and-forth
on $b^j$ at least twice)  if $i$ is divisible by $j$, and if so, $M$ accepts.  Clearly,  $M$'s
counter makes only one reversal.

It follows from the construction that if $i$ is a multiple of $j$ and
$i \neq j$, then
$fc_1^i$ would be a  reachable configuration in some accepting
computation.

Hence, $S(M) \cap fc_1^* = \{ fc_1^n \mid n$ is composite$\}$.

If $S(M)$ is in $\LFam(\NPCM)$, then $L' = S(M) \cap fc_1^* = \{fc_1^n \mid n \mbox{~is composite}\}$
is in $\LFam(\NPCM)$.  This is a contradiction, since the
Parikh map of $L'$ is not semilinear, but it is known that
the Parikh map of any $\NPCM$ language is semilinear \cite{Ibarra1978}.
\qed \end{proof}

For nondeterministic machines, only a unary alphabet is needed to obtain a similar result.
\begin{proposition} 
\label{finiteCrossing2NFA}
There is a non-finite crossing $2\NCM$ $M$ 
with one $1$-reversal
counter over a unary alphabet such that $S(M) \notin \LFam(\NPCM)$.
\end{proposition}
\begin{proof}
Construct an $M$ which first stores in the counter, a nondeterministically
chosen number $n$ and enters state $f$.  Then it changes state and checks
that $n$ is larger and a multiple of the length of the unary input.  As in Proposition \ref{finiteCrossing2DFA}, $S(M)$ is not in $\LFam(\NPCM)$.
\qed \end{proof}

The next result was already mentioned in Section \ref{sec:oneway},
but the proof appears here since it involves a proof using two-way
machines.
\begin{proposition}
There exists $M \in \NSCM$ with one $1$-reversal-bounded counter over a
unary alphabet such that $S(M) \notin \LFam(\NPCM)$.
\end{proposition}
\begin{proof}
Let $M$ be a $2\NCM$ with one $1$-reversal-bounded counter over a unary
language such that $S(M) \notin \LFam(\NPCM)$, which exists
by Proposition \ref{finiteCrossing2NFA}. Create an $\NSCM$ $M'$
with one counter, where $M'$ copies the input to the stack (using new states), and then simulates
$M$ on the stack contents and the counter.
Assume $S(M') \in \LFam(\NPCM)$. Let
$g$ be a gsm that erases the stack contents (keeping only the state and the counter), and $g$
does not map any words starting with any new state before the input is copied. Then $g(S(M')) = S(M) \in 
\LFam(\NPCM)$ since this family is closed under homomorphism, inverse homomorphism, and intersection with regular languages \cite{Ibarra1978}, and is therefore closed under
gsm mappings, a contradiction.
\qed \end{proof}

\begin{comment}

Finally, in contrast to Proposition \ref{finiteCrossing2NFA}, we have:
\begin{proposition}
If $M$ is a non-finite-crossing $2\DFA$ $M$ with one reversal-bounded
counter over a one-letter alphabet, then $S(M)$ is regular.
\end{proposition}
\begin{proof}

IAN: I think this is true, but not sure how to prove it yet.  Any idea?
I think that since the machine is deterministic, if the in put word
$a^n$  to $M$  (where the $\$$'s are the end markers) is long enough.
the incrementing (or decrementing) would be periodic.

\noindent

{\bf If we are not able to prove this, we can delete it later.}

\qed \end{proof}
\end{comment}

This subsection is concluded with a result that shows that 
for a particular two-way model of computation, the
store languages can be more complex than the languages
accepted.

\begin{proposition}
$~~~$
\begin{enumerate}
\item
If $M$ is a $2\NCM$ 
over a unary input alphabet,
then $L(M)$ is regular.
\item
There is a $2\NCM$ $M$ with one 1-reversal-bounded
counter over a unary input alphabet such that 
$S(M)$ is not semilinear (hence, $S(M)$ is not regular).
\end{enumerate}
\end{proposition}
\begin{proof}
Part 1 was shown in \cite{IbarraJiang}.
Part 2 follows from the language used in the proof of Proposition \ref{finiteCrossing2NFA}.
\qed \end{proof}
Hence, the languages accepted by the machines are all regular,
but the store languages are not even semilinear.

\subsection{Connections Between One-Way and Two-Way Machines}
\label{subsec:twowayoneway}

This subsection establishes some general connections between one-way and two-way machines.
First, a straightforward lemma is demonstrated to show
that store languages of one-way nondeterministic machines
are equivalent to those only accepting the empty word.

\begin{lemma} \label{empty1}
Let $\Omega_1, \ldots, \Omega_k$ be store types, and let $\MFam$ be the set of one-way nondeterministic $(\Omega_1, \ldots, \Omega_k)$-machines.
If $M \in \MFam$, then there exists $M' \in \MFam$ such that
$L(M') = \{\epsilon\}$ and $S(M') = S(M)$. Hence, the family
$\{S(M) \mid M \in \MFam\} = \{S(M) \mid M \in \MFam, L(M) = \{\epsilon\}\}$.
\end{lemma}
\begin{proof}
Construct $M'$ which, on $\epsilon$ input, guesses and simulates
the computation of $M$ on some input $x$  symbol-by-symbol. Since the sequence of
ways the store can change is the same as in $M$, then $M'$ must be an $(\Omega_1, \ldots, \Omega_k)$-machine
(i.e.\ in the definition of store types, all sequences of store instructions used in $M$ in accepting computations
are the same for $M'$, thereby being in the instruction language), and so $S(M) = S(M')$.
\qed \end{proof}
%This would be true for any type of nondeterministic machine defined in the appendix.
The above lemma is not true for deterministic machines $M$, since $S(M)$
may be infinite, but if $L(M')$ only accepts $\epsilon$, then $S(M')$ for any deterministic machine $M'$ is
always finite.

Next, a connection will be demonstrated between sets of one-way and two-way machines
of the same store type. The proposition involves two sets of machines with the same stores, where the first has a one-way input, and the second has a two-way input.
For example,  if ${\cal M}_1$ is the class of $\NPDA$s with $k$ reversal-bounded
counters, then ${\cal M}_2$ is the class of $2\NPDA$s with $k$ reversal-bounded
counters. It shows that the store languages for one-way machines are ``almost'' the same as two-way machines of the same type. The only difference is in the state.
\begin{proposition} \label{empty2}
Let $\Omega_1, \ldots, \Omega_k$ be store types,
let ${\cal M}_1$ be the set of one-way nondeterministic $(\Omega_1, \ldots, \Omega_k)$-machines and 
let ${\cal M}_2$ be the set of two-way nondeterministic $(\Omega_1, \ldots, \Omega_k)$-machines.
Then the following are true:
\begin{enumerate}
\item $\{S(M) \mid M \in \MFam_1\} = \{S(M) \mid M \in \MFam_1, L(M) = \{\epsilon\}\}
\subseteq \{S(M) \mid M \in \MFam_2, L(M) = \{\epsilon\}\}$.
\item For all $M_2 \in \MFam_2$ with $L(M_2)$ finite, there exists $M_1 \in \MFam_1$
with $L(M_1) = \{\epsilon\}$ and a homomorphism $h$ (that only can change the states)
such that $h(S(M_1)) = S(M_2)$.
\end{enumerate}
\end{proposition}
\begin{proof}
For item 1, from Lemma \ref{empty1}, 
$\{S(M) \mid M \in \MFam_1\} = \{S(M) \mid M \in \MFam_1, L(M) = \{\epsilon\}\}$.
Then, for every machine in the second set, a two-way machine can be
constructed (on epsilon input and thus the two-way head never
moves off end-markers) with the same store language.

For item 2, let $M_2 \in \MFam_2$.
For each $w \in L(M_2)$, there exists $S_w(M_2)$ consisting of all
words $x \in S(M_2)$ that can appear on the store in an accepting computation on input $w$.
Then $\bigcup_{w \in L(M_2)}S_w(M_2) = S(M_2)$.

Construct a machine $M_1$ in ${\cal M}_1$ as follows:
$M_1$ stores in its state a simulated state of $M_2$, a word $w \in L(M)$, and a position of $|w|$.
In the first move applied, $M_1$ guesses $w$, and simulates $M_2$ on $w$ by updating the state,
the stored input position, and the stores faithfully. As the sequences of store instructions
are identical, and sequences of valid instructions of one machine will have the corresponding
sequence in the other machine be valid. Finally, although the stores change identically
in accepting computations, the states of $M_1$ are different, as they contain also a word and a position.
But those can be transformed via a homomorphism $h$ that projects onto the simulated state.
Thus, $h(S(M_1)) = S(M_2)$.
\qed \end{proof} 

\begin{corollary}\label{onewaytwoway}
Let $\Omega_1, \ldots, \Omega_k$ be store types,
let ${\cal M}_1$ be the set of all one-way nondeterministic $(\Omega_1, \ldots, \Omega_k)$-machines and 
let ${\cal M}_2$ be the set of all two-way nondeterministic $(\Omega_1, \ldots, \Omega_k)$-machines,
and let $\LFam$ be a family closed under homomorphism. Then the following are equivalent:
\begin{enumerate}
\item
$\{S(M) \mid M \in \MFam_1\} \subseteq \LFam$,
\item
$\{S(M) \mid M \in \MFam_1, L(M) = \{\epsilon\}\} \subseteq \LFam$,
\item
$\{S(M) \mid M \in \MFam_2, L(M) \mbox{~finite}\} \subseteq \LFam$.
\end{enumerate}
\end{corollary}

As applications of the above, the following corollaries
to the results already shown in Section \ref{sec:oneway} are obtained:
\begin{enumerate}
\item
If $M$ is a $2\NPDA$ and $L(M)$ is finite, then $S(M)$ is regular.
\item
If $M$ is a $2\NTM$ with reversal-bounded read/write tape and $L(M)$ is
finite, then $S(M)$ is regular.
\item
If $M$ is a $2\NTM$ with reversal-bounded read/write tape and
reversal-bounded counters and $L(M)$ is finite, then $S(M)$ is in $\LFam(\NCM)$.
\end{enumerate}
Similar corollaries hold for the other machine models studied in
Section \ref{sec:oneway}.

The assumption that  $L(M_2)$ is finite in the above Proposition \ref{empty2}
is necessary.  Consider $2\DCA$, the set of two-way deterministic machines with 
an unrestricted counter (no reversal-bound).
\begin{proposition}
There is a $2\DCA$ $M$ 
which makes two sweeps on the input (left-to-right
and then right-to-left, where acceptance is on the left end-marker) and
makes only $O(\log n)$ reversals on the counter on input of size $n$ 
such that 
$S(M)$ is non-regular.  
\end{proposition}
\begin{proof}
Construct $M$ which, when given input $w$, operates as follows:
\begin{enumerate}
\item
$M$ makes a left-to-right sweep of the input $w$ and checks that
it is of the form   $$\rhd a^{i_1}b^{j_1}a^{i_2}b^{j_2} \cdots a^{i_k}b^{j_k}\lhd$$
for some $k \ge 1, i_1, \ldots, i_k, j_1, \ldots, j_k \ge 1$.
It uses the counter to check that $i_i = j_1, \ldots, i_k = j_k$.
At the end of this  process, the counter is zero.
\item
Then $M$ moves its input head left and increments the counter
to value $i_k  (= j_k)$ and enter a {\em unique} state $f$.  Thus the
configuration of the counter and state at this time is  $f c_1^{i_k}$.
The state $f$ is {\em only} entered at this time.
\item
Next, $M$ continues moving left checking that $j_{k-1} = i_k/2,
j_{k-2} = i_{k-1}/2, \ldots, j_1 = i_2/2 = 1$ and accepts.   (This is
possible because there are two copies of $i_k$ in each block.)
\end{enumerate}

\noindent
$S(M)$ is non-regular; otherwise
$S(M) \cap f c_1^+  = f \{c_1^{2^n} ~|~ n \ge 1 \}$
would be regular.
Clearly $M$ makes $O(\log n)$ reversals on the counter. 
\qed \end{proof}
Hence, the store languages of one-way $\DCA$s are regular by Proposition
\ref{NPDAStorage}, but two-way $\DCA$s are not.

Next, the store language of two-way and one-way
nondeterministic stack automata will be addressed. 
In \cite{KutribCIAA2016}, it was shown that the
store language of a one-way stack automaton is regular. Here,
an alternative simple proof of 
this result is provided by using the
general connections established between one-way automata and
two-way automata in Corollary \ref{onewaytwoway}, and an existing older result on two-way stack automata. 
In \cite{Stack2}, it was shown
that the set of all words that can appear in the store of a two-way
stack automaton $M$ on an input $w \in \Sigma^*$ (not in general over all words, but over only a single word), when $M$ ``falls
off'' the right end-marker of $w$, is a regular language
(this was used as a key step to showing all two-way stack languages
are recursive). This fact will be combined
with the results of this section to show that all store languages
of one-way nondeterministic stack automata are regular. Two technical lemmas
are required before a proof of the main result (essentially
used to convert the notation used in \cite{Stack2} to our
notation).
\begin{lemma}
\label{firsthalf}
Let $M$ be a two-way nondeterministic stack automaton. Then
$\{ q x \rw y \mid (q_0, \rhd \lhd,  Z_0\rw , 1) \vdash^* 
(q, \rhd \lhd, x \rw y,1)\} \in \LFam(\REG)$.
\end{lemma}
\begin{proof}
In \cite{Stack2}, it is shown that, for each word $\rhd w \lhd$,
and each $q\in Q$
then $\{xqy \mid (q_0, \rhd w \lhd,  Z_0\rw,1) \vdash^* (q,\rhd w \lhd, x \rw y, |w|+3)\}$
is a regular language. Then it is clear that, using the empty word, 
$\{q x \rw y \mid (q_0, \rhd  \lhd,  Z_0\rw,1) \vdash^* (q, \rhd \lhd , x \rw y,3)\}$ is regular.

Let $M'$ be a new two-way nondeterministic stack machine with state set $Q \cup Q'$, 
$Q' = \{q' \mid q \in Q\}$ (primed versions). Then $M'$ simulates $M$, but at any nondeterministically
chosen step, if 
the simulated $M$ is
in state $q$, $M'$ can nondeterministically switch to $q'$ and move the input head past
the right end-marker using a new state $q'$. Then 
$X = \{q' x \rw y \mid (q_0, \rhd  \lhd,  Z_0\rw,1) \vdash_{M'}^* (q', \rhd \lhd , x \rw y,3), q' \in Q'\}$ which
is regular. Let $h$ be a homomorphism that maps each $q' \in Q'$ to $q$ and fixes all other letters.
Indeed, $h(X) =  \{ q x \rw y \mid (q_0, \rhd \lhd,  Z_0\rw,1) \vdash_M^* 
(q, \rhd \lhd, x \rw y,1)\} \in \LFam(\REG)$.
\qed \end{proof}

\begin{lemma}
\label{secondhalf}
Let $M = (Q,\Sigma,\Gamma,\delta,q_0,F)$ be a two-way nondeterministic stack automaton. Then
$\{ q x \rw y \mid (q, \rhd \lhd,  x \rw y ,1) \vdash_M^* 
(q_f, \rhd \lhd, z,1), q_f \in F, z\in \Gamma^*\} \in \LFam(\REG)$.
\end{lemma}
\begin{proof}
In a standard proof that shows a one-way nondeterministic stack automaton is closed under
reversal, from an automaton $M$, another $M'$ is constructed that guesses the final stack
contents and pushes it while also guessing the position of the read head inside (using new states), guesses a final state of $M$, then simulates $M$ ``in reverse''; if $M$ pushes, $M'$ pops; if $M$ pops, $M'$  pushes, if $M$ moves left in the stack, $M'$ moves right, etc. The same construction works
for two-way nondeterministic stack automata on $\epsilon$ input.

Hence, from $M$,
let $M' = (Q',\Sigma,\Gamma,\delta',q_0',F')$ be a new two-way nondeterministic stack automaton constructed in this way. It does not ever move its input head, and
on a new initial state $q_0'$, nondeterministically
guesses a word $z$ and puts it on the stack, then on another new state $q_1'$, moves the read head
of the stack to an arbitrary position inside (thus guessing  $z_1\rw z_2$), then $M'$
nondeterministically switches to any final state of $M$. From
there, $M'$ simulates
$M$ in reverse. So, if $M$ moves right in the stack, then $M'$ moves
left, if $M$ moves left, then $M'$ moves right. If $M$
replaces the top of the stack symbol $x$ with $b_1 \cdots b_m, m \geq 1$, then $M'$ pops $b_m$ down
to $b_2$ (using states not in $Q$), then replaces $b_1$ with $x$. If
$M$ pops $x$, then $M'$ pushes $x$, etc. 

Then, 
$X= \{ q x \rw y \mid (q_0', \rhd \lhd,  Z_0\rw,1) \vdash_{M'}^* 
(q, \rhd \lhd, x \rw y,1),q \in Q'\}$  is regular by Lemma \ref{firsthalf}.
Furthermore, $X \cap Q \Gamma^*$ is regular (thus omitting configurations reached on any new states
is also regular since regular languages are closed under intersection).
This set is equal to $\{qx \rw y \mid (q_0', \rhd \lhd,  Z_0\rw, 1)
\vdash_{M'}^* (q_1', \rhd \lhd, z_1 \rw z_2, 1) \vdash_{M'}
(q_f, \rhd \lhd, z_1 \rw z_2, 1) \vdash_{M'}^* (q, \rhd \lhd, x \rw y, 1), q_f \in F, q \in Q\}$. Further,
this set is equal to
$\{ q x \rw y \mid (q, \rhd \lhd, x \rw y,1) \vdash_{M}^* 
(q_f, \rhd \lhd , z_1 \rw z_2,1), q_f \in F, z\in \Gamma^*\}$, which must
therefore be regular.
\qed \end{proof}

By intersecting the two regular languages in the previous two lemmas, the following is obtained:
\begin{proposition}
Let $M$ be a two-way nondeterministic stack automaton such that $L(M) = \{\epsilon\}$. Then $S(M)$ is regular.
\end{proposition}
\begin{proof}
From Lemmas \ref{firsthalf} and \ref{secondhalf}, and since regular
languages are closed under intersection,
$\{ q x \rw y \mid (q_0, \rhd \lhd,  Z_0\rw,1) \vdash^* (q, \rhd \lhd, x \rw y,1) \vdash^* (q_f,  \rhd \lhd, z,1), q_f \in F, z\in \Gamma^*\} \in \LFam(\REG)$.
\qed \end{proof}

From this,  from Corollary \ref{onewaytwoway}, and since the regular languages are closed under homomorphism, the following is obtained:
\begin{corollary}
If $M$ is a two-way nondeterministic stack automaton such that $L(M)$
is finite, then $S(M) \in \LFam(\REG)$.
\end{corollary}

\begin{corollary}
If $M$ is a one-way nondeterministic stack automaton, then $S(M) \in \LFam(\REG)$.
\end{corollary}

Corollary \ref{onewaytwoway} is also useful in other circumstances. For example, if a one-way machine
model has store languages in some family $\LFam$ that is closed under homomorphism and
$\LFam$ has a decidable emptiness problem, then the corresponding two-way model has its store language on a fixed word $w$ being in $\LFam$. By testing whether this store language is non-empty, this is
determining whether $w$ is accepted by the two-way machine. Hence, membership is decidable for two-way machines. Therefore, for all one-way models studied here where the store languages are in
$\LFam(\REG)$ or $\LFam(\NCM)$, membership in the corresponding two-way models is then
decidable.

\section{Applications to Right Quotient}

There are some nice applications of the results in this paper.
For example, it was shown in \cite{DLTJournalIJFCS} that it is decidable whether
the language accepted by a one-way reversal-bounded pushdown automaton is
dense (the set of subwords is equal to $\Sigma^*$). Furthermore, this problem is
also decidable for nondeterministic Turing machines with a one-way read-only
input tape and a reversal-bounded worktape \cite{DenseJALC} (using Proposition
\ref{TMReg} proven here).
Also, certain applications to problems in the area of verification and model checking are presented in \cite{StoreApplications}.
Another application is addressed here.

A general proof is exhibited whereby it is shown that any
deterministic automata class $\MFam$ obtained from so-called ``readable'' store types, 
where the nondeterministic machines with the same store types
only have regular store languages, 
then $\LFam(\MFam)$ is closed under right quotient with regular languages.
This is perhaps surprising since right quotient seems to be quite
difficult for deterministic machines.

\begin{definition}
Let $\Omega$ be a store type. Define $\Omega$ to be {\em readable} if the following are true:\begin{itemize}
\item $\Omega$ has stay instructions.
\item At any point, if the store contains $y$ say, it is possible to switch to a configuration where
 the store can be read one letter at a time, either from left-to-right (like a queue), or right-to-left (like a pushdown).
\end{itemize}
\end{definition}
The first condition is enforcing that it is possible to keep
the same store contents. For example, with a pushdown automaton,
it is always possible to replace the top of the pushdown $x$
with $x$, thereby keeping it the same. One could define a store type which is a pushdown
with only push and pop instructions (the size of the stack is not allowed to stay the same), and such
a store type would not be readable.

\begin{proposition}
Let $\Omega_1, \ldots, \Omega_k$ be readable store types.
Let $\MFam_N$ be the set of all one-way nondeterministic $(\Omega_1, \ldots, \Omega_k)$-machines, and let
$\MFam_D$ be the set of all one-way deterministic $(\Omega_1, \ldots, \Omega_k)$-machines.
If $\SFam(\MFam_N) \subseteq \LFam(\REG)$, then
$\LFam(\MFam_D)$ is closed under right quotient with regular
languages.
\end{proposition}
\begin{proof}
Let $M_1$ be a deterministic machine $M_1 \in \MFam_D$ with state set $Q$.
Let $M_2$ be a $\DFA$. 
A deterministic machine $M_5 \in \MFam_D$ will be built
accepting the right quotient of $L(M_1)$ with $L(M_2)$.

First, build a new intermediate 
{\bf nondeterministic} machine $M_3 \in \MFam_N$ 
with states $Q \cup Q' \cup Q''$ with $Q, Q', Q''$ being disjoint, and $Q'$ being primed versions of
states in $Q$ ($Q''$ described below).  It accepts the following language:
$$\{ w x \mid wx \in L(M_1), x \in L(M_2)\}.$$

Intuitively, $M_3$ simulates $M_1$, and at some nondeterministically guessed spot, starts
simulating $M_2$ in parallel using a second component simulating $M_2$ in the states.
Specifically, at the nondeterministically guessed spot, 
if it's in state $q$, it switches to state $q' \in Q'$, 
then to a state in $Q''$ (requiring the store contents to not change
between these configurations, which is possible by the first condition of the readable store type definition), then $M_3$ continues the simulation only using states from $Q''$ (with two components, the second component simulating $M_2$). Certainly, $M_3$ is nondeterministic as it needs to guess where to start
simulating $M_2$.

Next, construct the store language $S(M_3)$. It is regular by the assumption. In fact, only words of $S(M_3)$ that begin with $Q'$ are needed. 
Consider $S(M_3) \cap Q' \Gamma^*$, and build a $\DFA$ $M_4$ accepting this set.

Now build a new {\bf deterministic} machine $M_5 \in \MFam_D$ 
that operates as follows. It simulates $M_1$  on the input $w$ until it hits the right input end-marker. At that point, say $y$ is the contents of the
store, and it is in state $q$. 
First, assume that there is only one store which can be read from left-to-right (the store is readable).
Then read $q'$ in the store language $\DFA$ 
$M_4$ and see if $q' y$ is in the store language deterministically on the store. If using a store that reads from right-to-left, instead use a $\DFA$ accepting $S(M_4)^R$ instead of using $S(M_4)$.
Similarly, if using $k>1$ stores that are all readable (but the store language is still regular),
then $M_4$ is constructed to reverse the subwords from stores read from right-to-left.
In any of the cases, if $M_4$ accepts $q' y$, then $M_5$ accepts the input. 

Let $w \in L(M_5)$. Then reading $w$ in $M_5$ (upon consuming the last letter) takes it to some
configuration $q y$. Then $q' y \in L(M_4)$, and so 
$q' y$ is in the store language of $M_3$, which means that the machine $M_3$ can accept from this configuration. 
And the fact that primed states are being used to enforce that it is at the right spot of the store language ensures that from that point on, the remaining word is in $L(M_2)$.
Thus, there must be some $x$ 
such that $wx$ is in $L(M_1)$ and $x$ is in $L(M_2)$.  

Conversely, if $wx \in L(M_1)$ with $x \in L(M_2)$, then reading $w$
in $M_1$ takes it to some configuration $q y$. Then $q' y$ must be
in $L(M_4)$. Hence, by the construction of $M_5$, $w \in L(M_5)$.

Hence, $\LFam(\MFam_D)$ is closed under right quotient with regular languages.
\qed \end{proof}
In the proof above, if the store languages of machines in $\MFam_N$ can
be effectively constructed, then the machines accepting the right quotients can
also be effectively constructed.

The following classes are readable, and hence the languages are closed under right quotient
with regular languages: deterministic pushdown automata, deterministic one counter automata,
deterministic $k$-flip pushdown automata, and deterministic reversal-bounded queue automata.

For deterministic stack automata, checking stack automata, and variants of $\DTM$s, they are not
exactly readable, and the proof above does not completely apply, but can be adjusted.
With e.g.\ stack automata, when $M_5$ reaches the end of the input, it could be in read mode;
i.e.\ the store contents could be $\gamma = Z_b y_1 \rw y_2 Z^t$ where $y_2 \neq \epsilon$.
In this case, in order to read the stack contents from right-to-left (similarly with left-to-right)
to verify that $\gamma$ is in $M_4$, the position of the read head is lost. (In other words, it
is easy to verify that $Z_b y_1 y_2 \rw Z_t$ is in $L(M_4)$, but not $\gamma$.) 
For deterministic Turing machines, it is possible to mark the position of the read/write head to
make it verifiable.
For stack automata, a slightly more complicated construction is needed. First make $M_4$ a complete $\DFA$ and adjust the stack alphabet to be ordered
pairs, where the first component is an element from $\Gamma$, and the second component
is a state in $M_4$. Whenever $M_5$ simulates the pushing of a symbol of $M_1$, $M_5$ pushes this
as the first component, and for the second component, pushes the state of $M_4$ obtained from the state in the second
component of the previous topmost
symbol by reading the stack symbol pushed. Thus, if the stack contains $(b_0,p_0)\cdots (b_m,p_m)$,
$b_i \in \Gamma, p_i$ is a state of $M_4$, then for all $i$, reading $b_0 \cdots b_i$ in $M_4$ ends in state $p_i$. If a pop instruction occurs, then the state of $M_4$ is recoverable from the second component.
At the end of the input, if $M_5$ is at the top of the stack, then the state in the second component
immediately indicates whether the stack contents is in $L(M_4)$. If $M_5$ is
inside the stack with say $(b_0,p_0) (b_1,p_1) \cdots (b_i,p_m) \rw (b_{i+1},p_{i+1}) \cdots (b_m,p_m)$
on the stack, then $M_5$ simulates $M_4$ starting from $p_i$, and verifies that from there, reading
$\rw b_{i+1} \cdots b_m$ brings $M_4$ to a final state, thus verifying that its contents are in the store language.
(Note that the state sequence $p_{i+1}, \ldots, p_m$ was calculated without reading $\rw$ first, and
therefore is different than reading $\rw b_{i+1} \cdots b_m$) Hence, it is possible to verify that
$b_0 \cdots b_i \rw b_{i+1} \cdots b_m \in L(M_4)$. Therefore, the proof can be adjusted to work for stack
automata and checking stack automata as well.

This implies closure under right quotient with regular languages for several families.
\begin{corollary}
The following language families are closed under right quotient with regular languages: 
\begin{itemize}
\item deterministic stack languages \cite{DetStackQuotient},
\item deterministic checking stack languages,
\item deterministic $k$-flip pushdown languages,
\item deterministic pushdown automata \cite{GinsburgDPDAs},
\item deterministic one counter automata \cite{TAMCJournal},
\item deterministic reversal-bounded queue automata,
\item deterministic one-way read-only input Turing machines with a reversal-bounded worktape.
\end{itemize}
\label{quotientCors}
\end{corollary}

This does provide an alternate, much shorter and more general proof for stack and
pushdown automata. It also resolves an explicitly stated unsolved
open problem for $k$-flip pushdown automata \cite{kflipquotient}.
All others are, to our knowledge, also unknown.

It is worth noticing the tight relationship between store
languages and quotients. The intuition behind the closures
under right quotient of all the families in 
Corollary \ref{quotientCors} is that when the deterministic
machines reach the end of their inputs, they can verify that
their store contents are in the regular language constructed
from the store language of a very similar nondeterministic machine.
This same technique can even be true for non-regular store
languages. For example, a similar technique could be
used to show that $\DCM$ is closed under right quotient with 
$\NCM$. This is because when the $\DCM$ reaches the end
of its input, it only needs to verify that its store contents
are in another $\NCM$ language, and the store language of
an $\NCM$ language is in $\DCM$. So it can do this in parallel
with additional counters. However, in \cite{TAMCJournal}, a 
more general technique was used to show that
$\DCM$ is closed under right quotient with even
more general families such as $\NPDA$ and $\NPCM$. 

Note as well that not all deterministic families are closed
under right quotient with regular languages, as
$\DPCM$ is not \cite{TAMCJournal}. Indeed, the store
of a $\DPCM$ is not necessarily in $\DPCM$, so when such
a machine reaches the end of its input, there is not any
way to verify that its store contents are ``good'' by using
a store language within another $\DPCM$ machine.

\section{Space Lower Bounds for Non-Regular Store Languages of Turing Machines}
 \label{sec:bounds}
 
In this section, the lower bounds will be studied on the space
complexity of $\NTM$s and $\DTM$s for the store language not to be
regular. Here, $1\NTM$ ($1\DTM$) is used to denote a nondeterministic
(deterministic) Turing machine with a one-way read-only input and a Turing tape,
and $2\NTM$ ($2\DTM$) is used to denote a nondeterministic
(deterministic) Turing machine with a two-way read-only input and a Turing tape.
 
A configuration of $M$ is a tuple $(q, \rhd w\lhd , x, i)$,
where $q$ is a state, $w$ is the input with the input
head on the $i$'th position, and the worktape contains
string $x$ which includes the read/write head.

Let $M$ be any such Turing machine with either a one-way or two-way read-only input
and one read/write worktape (i.e., store) tape.  The following
two notions of $M$ being $s(n)$ space-bounded are used (see \cite{SublinearBook}):

\begin{enumerate}
\item
$M$ is strongly $s(n)$ space-bounded if, for any input $w$ of length $n$,
all computations on $w$ (accepting or not) use at most $s(n)$ space on
the worktape.
%\item
%$M$ is weakly $s(n)$ space-bounded if, for any input $w$ of length $n$
%that is accepted, there is an accepting computation on $w$ that uses at
%most $s(n)$ space on the worktape.
\item
$M$ is middle $s(n)$ space-bounded if, for any input $w$ of length $n$ that 
is accepted, all accepting computations on $w$  use at most $s(n)$ space.
\end{enumerate}

The following known results are needed:
\begin{proposition} \label{known}
$~~$
\begin{enumerate}
\item
$\log \log n$ is the lower bound for accepting non-regular
languages by strongly (middle respectively) space-bounded $2\NTM$s and $2\DTM$s.
\cite{hartmanis,HopcroftUllman}.
%\item
%$\log \log n$ is the lower bound for accepting non-regular
%languages by weakly space-bounded $2\NTM$s and $2\DTM$s.
%\cite{Alberts85}.
\item
$\log n$ is the lower bound for accepting non-regular
languages by strongly (middle respectively) space-bounded $1\NTM$s and $1\DTM$s.
\cite{hartmanis}.
%\item
%$\log \log n$ is the lower bound for accepting non-regular
%languages by weakly space-bounded $1\NTM$s \cite{Alberts85}.
%\item
%$\log n$ is the lower bound for accepting non-regular
%languages by weakly space-bounded $1\DTM$s \cite{Alberts85}.
%\item
%$\log \log n$ is the lower bound for accepting non-regular
%languages by middle space-bounded $2\NTM$s and $2\DTM$s.
%\item
%$\log n$ is the lower bound for accepting non-regular
%languages by middle space-bounded $1\NTM$s and $1\DTM$s.
\end{enumerate}
\end{proposition}

In addition to the usual notion of the store language of space-bounded
Turing machines, also the {\em strong store language} will be
considered which is the set of reachable configurations; that is, 
if $M$ is a $2\NTM$ ($1\NTM$, $2\DTM$, $1\DTM$),
the strong store language of M is  
$S^s (M) = \{qw  \mid $ there is computation of
$M$ (accepting or not) on some input of length $n$
that enters a configuration with state $q$ and $w$ on the worktape$\}$.

\begin{proposition} \label{result1}
If $M$ is a middle $s(n)$ space-bounded $2\NTM$
and $s(n)$ grows slower than $\log \log n$, then $S(M)$ is regular.
\end{proposition}
\begin{proof}
Construct a $2\NTM$ $M'$ which, given an input
$wqx$, where $w$ is over the input alphabet of $M$,
$q$ is a state, and $x$ is over the worktape alphabet of $M$ 
(assume that the state set and
alphabets are distinct) operates as follows:
\begin{enumerate}
\item
$M'$ simulates $M$ on $w$.
\item
At some point nondeterministically chosen, $M$
stops the simulation. Let the state and store
contents of $M$ (and, hence,
also of $M'$) at that time be $q'$ and $x'$. 
$M'$ converts $x'$  to $q'x' \# q' x'$, where
$\#$ is a new symbol.  (Thus $M'$
makes two copies of $q'x'$ separated by $\#$ with $x'$ marking the position of the read/write head).
\item
$M'$ then resumes the simulation of $M$ using only
the area to the right of $\#$ in the worktape.
\item
When $M$ accepts, $M'$ checks that $qx$ on the input
is identical to $q'x'$ on the worktape and  accepts.
\end{enumerate}

Clearly $M'$ is also $s(n)$ space-bounded,  
hence $L(M')$ is regular by Proposition \ref{known}, part 1.
Now, the strings in $L(M')$ are of the form $wqx$. 
A homomorphism deleting $w$ is then applied.  It follows 
that the strong store language is regular.
\qed
\end{proof}
Furthermore, given a Turing machine $M$ that is strongly
$s(n)$ space-bounded, one can build
$M'$ exactly like $M$ but with all states final, and 
$M'$ is middle $s(n)$ space-bounded and $S^s(M) = S(M')$.
Therefore:
\begin{corollary}
If $M$ is a strongly $s(n)$ space-bounded $2\NTM$
and $s(n)$ grows slower than $\log \log n$, then $S^s(M)$ is regular.
\end{corollary}

Next, it will be shown that the $\log \log n$ bound above is tight.
\begin{proposition} \label{result2}
There is a strongly  $\log \log n$ space-bounded
$2\DTM$ $M$ such that $S^s(M)$ is not regular.
\end{proposition}
\begin{proof}
Let $L = \{ x_1 \# x_2 \# \cdots  \# x_k \#  ~|~  k \ge 1, x_i \in 1\{0,1\}^*,
x _1  = 1, x_{i+1} = x_i + 1 \mbox{~for~}  1 \le i  < k, x_k = 1^m \mbox{~for some~} m \}$. The addition is binary number addition.
So, e.g., $1 \# 10 \#\ 11 \# 100 \# 101 \# 110 \# 111 \#$ is in $L$.
Construct a $2\DTM$ $M$ which, when given a string $w = x_1 \# x_2 \# \cdots  \# x_k \#$, verifies
that $x_1 = 1, x_k = 1^m$ for some $m \geq 1$, each $x_i$ starts with $1$, and also verifies that each $x_i+1 = x_{i+1}$. To do the latter, $M$ uses the worktape to keep a binary counter referring to a bit position of each string $x_i$. The counter starts
at $1$, then it compares the last bit of $x_i$ to $x_{i+1}$, then it increases the counter by $1$, and
compares the second last bit of $x_i$ to $x_{i+1}$, etc. It is clear that this counter can grow
as the large as the length of the longest $x_i$. As the counter is in binary, this requires approximately
$\log |x_i|$ bits of space.
When $M$ determines that $w$
is in $L$, the worktape will have $m$ on its worktape in binary; call this string $b(m)$.
$M$ then transforms $b(m)$ to $b(m) \#b(m)$ and enters state $f$. 
Then $S^s(M) \cap f \{0,1\}^+ \# \{0,1\}^+ = \{ fw \# w ~|~ w \in 1 \{0,1\}^*\}$,
which is not regular.  Hence $S^s(M)$ is not regular. 

%In order to accept, $k$ must be approximately a power of $2$,
Clearly, on input longer than $k$, $|x_k|$ is approximately $\log k$,
and the worktape is approximately the size of $\log |x_k|$. Thus, $M$ is strongly
$\log \log n$ space-bounded.
\qed
\end{proof}
Hence, the following is immediate:
\begin{corollary}There is a middle  $\log \log n$ space-bounded
$2\DTM$ $M$ such that $S(M)$ is not regular.
\end{corollary}

Turning now to one-way machines:

\begin{proposition} \label{result4}
If $M$ is a middle $s(n)$ space-bounded $1\NTM$
and $s(n)$ grows slower than $\log n$, then $S(M)$ is regular.
\end{proposition}
\begin{proof}
The proof is the same as the proof of Proposition \ref{result1}
using Proposition \ref{known}, part 2, and noting that the $M'$
constructed in that proof would also be one-way if $M$ is one-way.
\qed
\end{proof}
\begin{corollary}
If $M$ is a strongly $s(n)$ space-bounded $1\NTM$
and $s(n)$ grows slower than $\log n$, then $S^s(M)$ is regular.
\end{corollary}

The next result shows that Proposition \ref{result4} is tight.

\begin{proposition} \label{result5}
There is a strongly $\log n$ space-bounded
$1\DTM$ $M$ such that $S^s(M)$ and $S(M)$ are not regular.
\end{proposition}
\begin{proof}
Let $L = \{a^n b^n ~|~ n \ge 1 \}$.
Construct a strongly $\log n$ space-bounded 
$1\DTM$ to accept $L$.  $M$ when given an input
$a^n b^m$, first reads $a^n$ and stores $n$
in binary, say $x$, on the worktape. Then
$M$ converts $x$ to $x \# x$. Next, $M$ reads
$b^m$ while decrementing the second $x$ on the worktape
to check that $m=n$.  Finally, $M'$ converts
the worktape to $x \# x$ and accepts in state $f$. 
Clearly, $S^s(M) \cap f (0+1)^+ \# (0+1)^+
= \{ f x \# x ~|~ x \in 1 \{0,1\}^* \}$
is not regular.  Hence, $S^s(M)$ is not regular. Making all states
final then gives the same result for $S(M)$.
\qed
\end{proof}
\begin{corollary}
There is a middle $\log n$ space-bounded
$1\DTM$ $M$ such that $S(M)$ is not regular.
\end{corollary}

\section{Conclusions and Future Directions}

Store languages are studied in a general fashion, by varying the types
of stores used. Certain specific models are studied such as 
nondeterministic Turing machines with a one-way read-only input
tape and a reversal-bounded read/write worktape, and it is shown that
all store languages are regular. Similarly, all store languages of $k$-flip pushdown automata are regular. Then it is shown that store languages of one-way 
nondeterministic, and one-way deterministic machines coincide,
when using the same store types. Similarly, these coincide with two-way
machines that accept finite languages over the same store types after applying a homomorphism.
One application of store languages is presented here. If there is a one-way
nondeterministic model with readable store types that only has regular store
languages, then the languages accepted by deterministic machines with
the same store types are closed under right quotient with regular languages. This resolves several open problems in the literature.
This type of result is only possible by studying store languages
in the general fashion done here. Lastly, space-bounded Turing machines
are studied, and lower bounds are given to have non-regular store
languages.

There are many other machine models in the literature that have yet to have their store language studied. The store languages of $\NPCM$ will be considered in a follow-up paper. Also, the pushdown hierarchy
is of interest \cite{IteratedStack}.
We also believe that
there are many other applications of store languages, similar
to the result on right quotient studied here.

\section*{Acknowledgements}
We thank the reviewers for their comments that improved the presentation of our results.

\bibliography{undecide_refs}{}
\bibliographystyle{elsarticle-num}

%\section*{Appendix}

\end{document}